\documentclass[12pt, draftclsnofoot, onecolumn]{IEEEtran}

\usepackage[ruled,vlined]{algorithm2e}
\usepackage{float}
\usepackage{amsmath}
\usepackage{subfigure}
\usepackage{graphicx}
\usepackage{caption}
\usepackage{cite}
\usepackage{amssymb}
\usepackage{ntheorem}
\usepackage{epstopdf}

\newtheorem{definition}{Definition}
\newtheorem{remark}{Remark}
\newtheorem{lemma}{Lemma}
\newtheorem*{proof}{Proof}

\newtheorem{theorem}{Theorem}
\newtheorem{proposition}{Proposition}

\begin{document}
\title{UAV Deployment, Device Scheduling and Resource Allocation for Energy-Efficient UAV-Aided IoT Networks with NOMA}
%

\author{
Jingjing\ Zhao, Kaiquan Cai, and Yanbo Zhu

\thanks{J.~Zhao is with the Research Institute for Frontier Science, Beihang University, Beijing 100191, China (e-mail: jingjingzhao@buaa.edu.cn). }
\thanks{K.~Cai is with the School of Electronic and Information Engineering, Beihang University, Beijing 100191, China (e-mail: caikq@buaa.edu.cn). }
\thanks{Y.~Zhu is with the School of Electronic and Information Engineering, Beihang University, Beijing 100191, China. He is also with the Aviation Data Communication Corporation, Beijing 100191, China (e-mail: zhuyanbo@buaa.edu.cn).}
}
\maketitle

\pagestyle{empty}
\thispagestyle{empty}

\begin{abstract}
This article investigates the energy efficiency issue in non-orthogonal multiple access (NOMA)-enhanced Internet-of-Things (IoT) networks, where a mobile unmanned aerial vehicle (UAV) is exploited as a flying base station to collect data from ground devices via the NOMA protocol. With the aim of maximizing network energy efficiency, we formulate a joint problem of UAV deployment, device scheduling and resource allocation. First, we formulate the joint device scheduling and spectrum allocation problem as a three-sided matching problem, and propose a novel low-complexity near-optimal algorithm. We also introduce the novel concept of `exploration' into the matching game for further performance improvement. By algorithm analysis, we prove the convergence and stability of the final matching state. Second, in an effort to allocate proper transmit power to IoT devices, we adopt the Dinkelbach's algorithm to obtain the optimal power allocation solution. Furthermore, we provide a simple but effective approach based on disk covering problem to determine the optimal number and locations of UAV's stop points to ensure that all IoT devices can be fully covered by the UAV via line-of-sight (LoS) links for the sake of better channel condition. Numerical results unveil that: i) the proposed joint UAV deployment, device scheduling and resource allocation scheme achieves much higher EE compared to predefined stationary UAV deployment case and fixed power allocation scheme, with acceptable complexity; and ii) the UAV-aided IoT networks with NOMA greatly outperforms the OMA case in terms of number of accessed devices.
\end{abstract}


\section{Introduction}
In the past several years, the use of unmanned aerial vehicles (UAVs) as flying communication platforms to boost the capacity and coverage of current wireless networks has attracted fast-growing interests \cite{7470933, 7317490, 8432474}. Different from terrestrial wireless communications, UAV-aided networks possess many appealing advantages including high mobility, flexible deployment, low cost, and line-of-sight (LoS) aerial-to-ground (A2G) links \cite{7876852}. Therefore, UAVs are expected to bring in promising gains to numerous use cases in next generation wireless networks. Particularly, the Internet-of-Things (IoT) consisting of a large scale of small devices (e.g., sensors, health monitors) is restricted by short distance transmission due to the energy constraints \cite{5741148}. In this case, UAVs can be despatched to collect data from IoT devices and transmit it to other devices or the data center out of the communication range. On the one hand, the maneuverability enables the dynamic adjustment of UAVs' positions to best suit the communication environment. On the other hand, the presence of LoS connections introduces better channel conditions, thereby improving the transmission efficiency and reducing the transmit power of IoT devices. 

These superiorities have inspired a proliferation of recent studies on the new research diagram of joint trajectory planning and resource allocation in UAV-aided wireless networks \cite{6863654, 8038869, 8247211, 8974403}. In particular, in \cite{6863654}, an analytical approach was presented to optimize the altitude of UAVs to provide maximum radio coverage on the ground. The authors in \cite{8038869} proposed a novel framework for efficiently deploying and moving UAVs to collect data from ground IoT devices.  The work in \cite{8247211} optimized the minimum throughput over all ground users by the joint UAV's trajectory design, multiuser communication scheduling, and power control. In \cite{8974403}, a robust resource allocation algorithm was designed for multiuser downlink multiple-input-single-output UAV communication systems with considering various uncertainties. 

In the IoT networks, massive connectivity is required to support large number of devices in various scenarios with limited spectrum resources \cite{5741148}. To leverage the spectrum resources more efficiently, non-orthogonal multiple access (NOMA) has been envisioned to be a revolutionizing technique for its potential to enhance spectrum efficiency by allowing multiple users simultaneous transmission in the same resource block (RB) \cite{7842433}. More specifically, the fundamental concept of NOMA is to facilitate the access of multiple users in a new dimension-power domain, by means of superposition coding (SC) and successive interference cancellation (SIC) at the transmitter and receiver side, respectively \cite{6868214, ding2016}. It is worth mentioning that due to the employment of superposition coding transmission scheme, the power allocation is an eternal problem to be investigated in NOMA, especially in multiple subchannels/subcarriers/clusters scenarios. Somewhat related power allocation and subchannel/subcarrier/cluster assignment problems have been studied in the context of NOMA \cite{lei2016power, 7954630, 8014491}. In \cite{lei2016power}, with formulating NOMA resource allocation problems under several practical constraints, the tractability of the formulated problem was analytically characterized. A novel resource allocation design was investigated for NOMA-enhanced heterogeneous networks in \cite{7954630} to maximize the system sum rate while guaranteeing users' fairness. In \cite{8014491}, the joint power and subchannel assignment problem was discussed in device-to-device communications with NOMA protocol. 

\subsection{Motivation and Contributions}
To reap the benefits of NOMA in terms of massive connectivity, integration of NOMA into UAV-based wireless networks has attracted some research contributions recently \cite{8663350, 8482444, 8918266, 8685130, 9123495, 8700188}. More particularly, in \cite{8663350}, the placement and power allocation were jointly optimized to improve the performance of the NOMA-UAV network. The optimization of radio resource allocation, decoding order of NOMA process as well as UAV's placement was investigated in \cite{8482444} to maximize sum rate. In \cite{8918266}, the joint placement design, admission control and power allocation was studied for the NOMA-based UAV downlink system to maximize number of connected users with satisfied quality-of-service (QoS) requirements. In \cite{8685130}, the authors investigated the joint trajectory design and resource allocation algorithms to maximize the minimum average rate among ground users for UAV communication systems with NOMA. It is worth noting that the aforementioned work in \cite{8663350, 8482444, 8918266, 8685130} concentrated on the downlink NOMA-UAV system, while the uplink scenario plays a vital role in UAV-aided IoT networks for data collection \cite{7470933}. Accordingly, in \cite{9123495}, a synergetic scheme for UAV trajectory planning and subslot allocation was proposed to maximize the uplink average achievable sum rate of IoT terminals. The authors in \cite{8700188} aimed to maximize the system capacity by jointly optimizing the subchannel assignment, the uplink transmit power of IoT nodes and the flying heights of UAVs. However, both in \cite{9123495} and \cite{8700188}, the UAV deployment and resource allocation schemes were for throughput maximization. With energy efficiency (EE) being a major concern in IoT networks due to constrained energy, studying the EE maximization is rather important. 

Different from previous works, in this paper, we investigate the EE of an uplink UAV-aided IoT networks with integrated NOMA to support a large number of IoT devices. Particularly, we consider the setting of a single UAV flying above the target area while stopping at a number of locations to collect data from ground devices. Thanks to the UAV's flexible deployment, determining the UAV's optimal positions and scheduling devices to transmit when the UAV is located at the best locations can effectively enhance the network capacity and reduce power consumption. Moreover, devices are allowed to transmit to the UAV on the same subchannel simultaneously via the NOMA protocol, as thus proper subchannel and power allocation can significantly improve EE. Consequently, the main contributions of this work are summarized as follows: 
\begin{itemize}
	\item By jointly optimizing UAV deployment, device scheduling and resource allocation, we aim to maximize EE in the UAV-aided IoT networks with NOMA. To obtain the optimal solution, we propose a staged optimization approach by decoupling the formulated problem. 
	\item Assuming given UAV deployment strategy, we model the joint problem of device scheduling and subchannel allocation as a three-sided matching among IoT devices, UAV's stop points (SPs) and subchannels. By converting the three-sided matching to a many-to-one two-sided matching problem between IoT devices and subchannel-SP (SS) units, we propose a near-optimal matching algorithm (JDSSA-1) with low complexity to maximize EE. To further improve the performance of JDSSA-1, we introduce the concept of `exploration' into the matching game and propose a novel algorithm JDSSA-2, where irrational swap decisions are enabled with a small probability to explore the potential matching states. We analyze the proposed algorithms in terms of stability, convergence, complexity and optimality.
	\item Regarding the power allocation for IoT devices associated with the same SS units, the formulated problem is proved to be pseudo-concave. We propose a low-complexity algorithm by applying the Dinkelbach's  algorithm and obtain the optimal power allocation results.
	\item We formulate the UAV deployment issue as a disk covering problem, and propose a simple but effective approach to find the optimal number and locations of SPs such that the UAV can completely cover the target area via LoS links, thereby improving transmission efficiency and reducing power consumption of IoT devices. 
	\item We propose a joint UAV deployment, device scheduling and resource allocation algorithm JUDDSRA. Through extensive simulation, we demonstrate that JUDDSRA can achieve much higher EE compared to the benchmarks with stationary UAV deployment and fixed power allocation. It is also shown that the UAV-aided IoT networks with NOMA achieve much higher number of accessed IoT devices compared to the OMA case. 
\end{itemize}    

\subsection{Organization and Notations}
The rest of this paper is organized as follows. In Section II, the system model of UAV-aided IoT networks with NOMA is presented. Section III formulates the EE maximization problem, and analyzes the computational complexity. In Section IV, we study the joint device scheduling and subchannel allocation problem. Section V addresses the power allocation issue. Section IV discusses the UAV deployment solution. Numerical results are presented in Section VII, which is followed by conclusions of this work in Section VIII.

The notations of this paper are shown in Table I.

\begin{table}[htbp]
\caption{Notations}
\begin{center}
\begin{tabular}{p{0.17\textwidth}|p{0.57\textwidth}}
\hline

\hline
$\mathcal{M}$, $\mathcal{N}$, $\mathcal{K}$ & Set of IoT devices, subchannels, and SPs\\
\hline
$M$, $N$, $K$ & Number of IoT devices, subchannels, and SPs\\
\hline
$R$ & Target area radius\\
\hline
$H$ & UAV altitude\\
\hline
$\boldsymbol{s}_k$ & Position of the $k$-th SP\\
\hline
$g_{m,u,k}$ & Channel power gain between IoT device $m$ and the UAV at SP $k$\\
\hline
$d_{m,u,k}$ & Distance between IoT device $m$ and the UAV at SP $k$\\
\hline
$\eta$ & Unit power gain at the reference distance $d_0 = 1$ m\\
\hline
$\mathcal{M}_{n,k}$ & Set of IoT devices scheduled to transmit to UAV at SP $k$ over subchannel $n$\\
\hline 
$\mathcal{\zeta}$ & Additive white Gaussian noise\\
\hline
$\gamma$ & Signal-to-interference-plus-noise ratio\\
\hline
$R_{LoS}$ & UAV LoS coverage radius\\
\hline
$a_{m,n}$ & Subchannel allocation indicator\\
\hline
$b_{m,k}$ & Device scheduling indicator\\
\hline
$\theta ^{thr}$ & Minimum elevation angle of UAV for LoS coverage\\
\hline
$P_{max}$ & Maximum transmit power of IoT device $m$\\
\hline

\hline
\end{tabular}
\end{center}
\label{table:1}
\end{table}
\section{Network Model}
As shown in Fig.~\ref{fig:system_model}, we consider a UAV-aided IoT communication system consisting of $M$ IoT devices, denoted as $\mathcal{M} = \{ 1, ..., m, ..., M \}$. Without loss of generality, we consider a circular area with radius $R$, centered at $(0, 0, 0)$. The UAV acts as an on-demand aerial access point (AP) to promote the efficiency of data collection. The UAV flies right above the coverage area, at a fixed altitude $H$. The UAV sequentially flies to given positions and hovers to collect data. Let $\mathcal{K} = \{1, ... k, ..., K\}$ denote the set of SPs and $\boldsymbol{s}_{k} = \{ x_k, y_k, H \}$ denote the position of UAV at the $k$-th SP. During the UAV flight, IoT devices connect to the UAV when being scheduled for data transmission. Let $\mathcal{N} = \{ 1, ..., n, ..., N \}$ denote the set of orthogonal subchannels in the system. Due to the dominance of LoS propagation, we assume that the communication links from IoT devices to the UAV are frequency-flat over the subchannels for simplicity. We assume that $M > NK$, and each device has access to only one subchannel. To improve the spectrum efficiency and guarantee that all the IoT devices can be served, multiple devices can share the same subchannel for data transmission via the NOMA protocol. 
In our model, we consider centralized network in which the locations of ground devices and the UAV are known to a control center located at a central cloud server. The cloud server will determine the UAV deployment, device scheduling and resource allocation in the network.

\begin{figure}
  \centering
  \includegraphics[width= 3.3in]{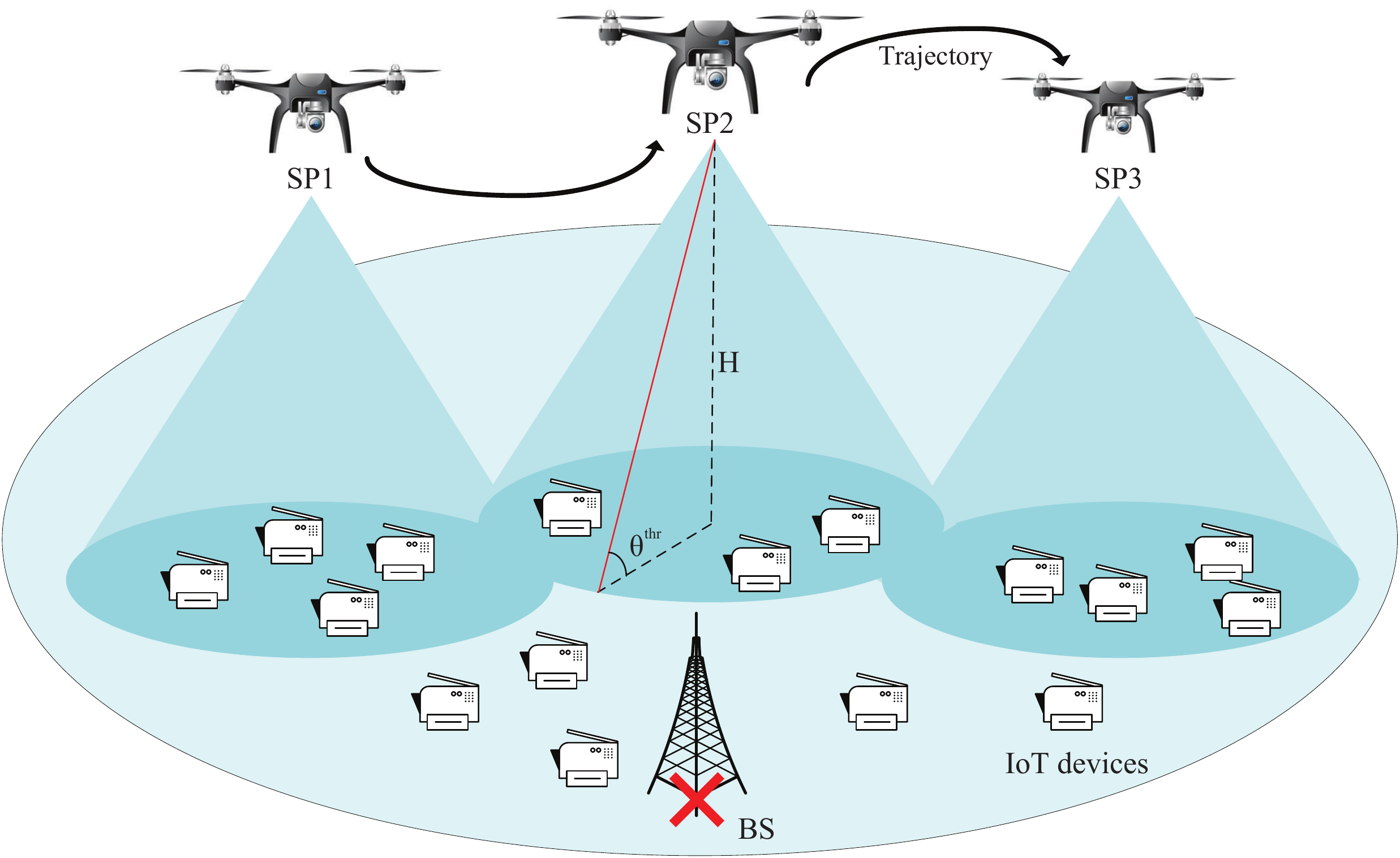}
  \caption{Illustration of the UAV-aided IoT networks scenario.}
  \label{fig:system_model}
\end{figure}

Unlike terrestrial communications, within the coverage of UAV, the LoS links are preferred and are easily obtained. To exploit the benefits of LoS links, in this work, we only consider the LoS communications between UAV and ground devices, according to \cite{9} as well as some recent measurement results \cite{10} \cite{11}. Let $\theta$ denote the elevation angle of UAV. According to the first Fresnel Zone Theory \cite{2}, the smaller the $\theta$ is, the lower the probability for the LoS communication. In order to express the LoS range, at altitude $H$, UAV's minimum elevation angle for LoS coverage is denoted as $\theta^{thr}$. Thus, IoT device $m$ with location $(x_m, y_m, 0)$ can transmit signal to the UAV at SP $k$ via LoS link if the following constraint is satisfied:
\begin{equation}
\label{eq:los-constraint}
\sqrt{(x_m - x_k)^2 + (y_m - y_k)^2} \leq H\times \cot \theta^{thr}.
\end{equation}
Inside the LoS coverage, the channel power gain between IoT device $m$ and UAV at SP $k$, i.e., $g_{m, u, k}$, follows the free-space path loss model, which can be quantified by
\begin{equation}
\label{eq:a2g_channel}
g_{m, u, k} = \frac{\eta}{(d_{m, u, k}) ^2},
\end{equation}
where $\eta$ denotes the unit power gain at the reference distance $d_0 = 1$ m, and $d_{m,u,k}$ denotes the distance between device $m$ and the UAV at SP $k$ with $d_{m,u,k} = \sqrt{(x_m - x_k)^2 + (y_m - y_k)^2 + H^2}$. Note that unlike conventional fading channels with randomness, the channels from devices to the UAV are dominated by LoS link, which makes the time-varying channels to be predictable. As a result, the centralized UAV deployment, device scheduling and resource allocation scheme can be proactively designed to maximize network EE. 

For uplink NOMA, multiple devices are allowed to transmit on the same subchannel, while the UAV needs to adopt SIC to demodulate the targeted messages. On the UAV side, we assume that the decoding order is always from the devices with better channel conditions to the devices with worse channel conditions, on the purpose of saving transmission power for worse devices to compensate the severe channel attenuation~\cite{7542118}. Assume that a set of devices $\mathcal{M}_{n, k}$ are scheduled to transmit to UAV at SP $k$ over subchannel $n$, then the received signal at UAV on subchannel $n$ is given by
\begin{equation}
y_{n, k} = \sum _{m = 1}^{\left|\mathcal{M}_{n, k}\right|}\sqrt{p_{m, n, k}g_{m, u, k}}x_{m, n, k} + \zeta _{n, k},
\end{equation}
where $p_{m, n, k}$ is the transmit power of device $m$ on subchannel $n$. In addition, $x_{m, n, k}$ denotes the signal transmitted from device $m$ over subchannel $n$, with $\mathbb{E}(\left|x_{m, n, k}\right|) = 1$. $\zeta _{n, k}$ represents the additive white Gaussian noise (AWGN) at the UAV with variance $\sigma^2$. Assume that the devices in $\mathcal{M}_{n, k}$ are arranged in descending order with respect to channel conditions, i.e., $g_{1, u, k} > ... > g_{\left|\mathcal{M}_{n, k}\right|, u, k}$, then we can get the SINR expression at the $m$-th device in $\mathcal{M}_{n, k}$ as follows:
\begin{equation}
\label{eq:sinr_mnk}
\gamma _{m, n, k} = \frac{p_{m, n, k}g_{m, u, k}}{\sum _{i = m + 1}^{\left|\mathcal{M}_{n, k}\right|}p_{i, n, k}g_{i, u, k} + \sigma ^2}.
\end{equation}

\section{Problem Formulation}
In this section, we first give definition on network overall EE, and then we formulate the joint UAV deployment, device scheduling and resource allocation optimization problem for UAV-aided IoT uplink transmission system with NOMA.
\subsection{Network EE}
The EE is independent among orthogonal time/frequency resources. As such, we consider the sum EE for subchannels over different UAV SPs in this treatise. It has been shown in \cite{5} that the power consumption of a transmitter mainly consists of two parts: the dynamic power consumed in the power amplifier (PA) for transmission and the static power consumed for circuits. Thus, the total power consumption on subchannel $n$ over SP $k$ can be expressed as
\begin{equation}
\label{eq:total_power}
P_{n, k} = \sum_{m = 1}^{M} a_{m,n}b_{m, k} (p_{m, n, k} + P_0), 
\end{equation}
where $P_0$ denotes the static power consumption for circuits of IoT devices, which is caused by filters, frequency synthesizer, analog-to-digital converters, etc..  $a_{m, n}$ is the subchannel allocation indicator, i.e., if subchannel $n$ is allocated to device $m$, $a_{m, n} = 1$; otherwise, $a_{m, n} = 0$. $b_{m, k}$ is the device scheduling indicator, i.e., if device $m$ transmits data to UAV at SP $k$, $b_{m, k} = 1$; otherwise, $b_{m, k} = 0$.  Therefore, the sum rate for subchannel $n$ over SP $k$ is expressed as
\begin{equation}
\label{eq:sum_rate}
R_{n, k} = \sum_{m = 1}^{M} a_{m,n}b_{m, k} \log _2(1 + \gamma _{m, n, k}).
\end{equation} 
According to (\ref{eq:total_power}) and (\ref{eq:sum_rate}), the network EE is given by
\begin{equation}
\label{eq:sum_ee}
EE = \sum_{n = 1}^N\sum_{k = 1}^K \frac{R_{n, k}}{P_{n, k}}.
\end{equation}

\subsection{Optimization Problem Formulation}
Let $V \triangleq \{\{a_{m, n}\}, \{b_{m, k}\}, \{p_{m, n, k}\}, \{\boldsymbol{s}_k\}\}$. Specifically, the joint UAV deployment, device scheduling and resource allocation problem can be formulated as the following:
\begin{subequations}
\label{eq:optimization_problem}
\begin{equation}
\mathop {\max }_{V} EE,
\end{equation}
\begin{equation}
s.t. \ \ a_{m, n}\in \left\{0,1\right\}, \ \ \forall m, n,
\label{eq:d}
\end{equation}
\begin{equation}
b_{m, k}\in \left\{0,1\right\}, \ \ \forall m, k,
\label{eq:e}
\end{equation}
\begin{equation}
\sum _{n = 1}^N a _{m, n} \leq 1, \ \ \forall m ,
\label{eq:f}
\end{equation}
\begin{equation}
\sum _{k = 1}^K b_{m, k} \leq 1, \ \ \forall m,
\label{eq:g}
\end{equation}
\begin{equation}
\sum _{m = 1}^M a_{m,n}b_{m, k} \leq q, \ \ \forall n, k,
\label{eq:quota_constraint}
\end{equation}
\begin{equation}
\label{eq:los_constraint_pf}
b_{m,k}\sqrt{(x_m - x_k)^2 + (y_m - y_k)^2} \leq H \cot \theta^{thr}, \forall m, k,
\end{equation}	
\begin{equation}
\label{eq:power_geq_0}
\sum _{n = 1}^N\sum _{k = 1}^K a_{m, n}b_{m, k}p_{m,n,k} \geq P_{min}, \ \ \forall m,
\end{equation}
\begin{equation}
\sum _{n = 1}^N\sum _{k = 1}^K a_{m, n}b_{m, k}p_{m, n, k} \leq P_{max}, \ \ \forall m.
\label{eq:c}
\end{equation}
\end{subequations}
Constraint (\ref{eq:d}) and (\ref{eq:e}) show that the values of $a_{m, n}$ and $b_{m, k}$ should be either $0$ or $1$. Constraint (\ref{eq:f}) guarantees that only one subchannel can be allocated to each device, while constraint (\ref{eq:g}) guarantees that each device can transmit at one of UAV's SPs. Constraint~(\ref{eq:quota_constraint}) ensures that at most $q$ IoT devices can be allocated to each subchannel at each of UAV's SPs, which is to reduce the implementation complexity and the co-channel interference. Constraint (\ref{eq:los_constraint_pf}) is imposed to restrict that the IoT devices are scheduled to transmit when they are within the UAV's LoS coverage. (\ref{eq:power_geq_0}) is the minimum transmit power constraint for guaranteeing the fairness among IoT devices. Constraint (\ref{eq:c}) gives the upper bound of transmit power of IoT devices. 

The formulated problem is neither convex nor quasi-convex due to the sum-of-ratios objective function, products of optimization variables in (\ref{eq:power_geq_0})-(\ref{eq:c}), and the binary optimization variables $a_{m, n}$ and $b_{m, k}$. Therefore, there is no systematic and computational efficient approach to solve this problem optimally. Furthermore, an exhaustive search for all possible solutions to obtain the optimal EE has the exponential complexity which is unpractical especially in a dense network. Therefore, in Section IV, \ref{sec:subchannel_allocation} and \ref{sec:power_allocation}, we decouple the formulated problem into three sub-problems: 1) joint device scheduling and subchannel allocation; 2) power allocation for IoT devices; and 3) UAV deployment. We develop low-complexity and effective algorithms for each of the sub-problems to obtain a suboptimal solution for problem (\ref{eq:optimization_problem}). 

\section{Joint Device Scheduling and Subchannel Allocation Based on Matching Game}
\label{sec:subchannel_allocation}
In this section, assuming given UAV deployment and fixed power allocation, we focus on the joint device scheduling and subchannel allocation problem. To describe the dynamic matching among devices, subchannels and UAV at different SPs, we consider the joint device scheduling and subchannel allocation problem as a three-sided matching game. IoT devices, subchannels and UAV at different SPs act as three sets of players and interact with each other to maximize the network EE. We adopt the matching theory \cite{roth1992two,gu2015future}, which provides mathematically tractable and low-complexity solutions for the combinatorial problem of matching players in distinct sets \cite{manlove2013algorithmics}. Subsequently, we propose two efficient algorithms to solve this problem.
\subsection{Three-Sided Matching Game Formulation}

In our matching model, IoT devices, subchannels and UAV at different SPs are three sets of agents who aim to maximize their own profile. Each agent has the order of preferences over the pairs of other agents that they are ready to form triples with. To proceed with proposing the efficient three-sided matching algorithm, we first introduce some notations and basic definitions for the matching model. 

\begin{definition}
In the three-sided matching model, if IoT device $m$, subchannel $n$ and UAV at SP $k$ are matched together, they three compose a \emph{matching triple}, denoted as $\boldsymbol{T} = \left(m, n, k\right)$.
\end{definition}

As shown in (\ref{eq:f}) - (\ref{eq:quota_constraint}), on one hand, each device $m$ can transmit data to UAV at one of UAV's SPs by occupying one subchannel. On the other hand, at most $q$ IoT devices can be allocated to one subchannel at one SP. To better describe the matching process among IoT device, subchannels and SPs, we introduce the concept of ``\textit{subchannel-SP (SS) unit}", i.e., $(n, k) \in \mathcal{N}\times \mathcal{K}$, as the unit between subchannels and SPs. We can then define the three-sided matching model as the following:
\begin{definition}
\label{def:three_sided_matching}
In the three-sided matching model among $\mathcal{M}$, $\mathcal{N}$ and $\mathcal{K}$, a matching $\Omega$ is the set of matching triples which satisfies:

1) $\Omega(m) \in \mathcal{N}\times \mathcal{K}$;

2) $\left|\Omega(m)\right| \leq 1$;

3) $\Omega(n,k) \subseteq \mathcal{M}$;

4) $\left|\Omega(n, k)\right| \leq q$;

5) $\Omega(m) = (n, k) \Leftrightarrow m \in \Omega(n, k)$.
\end{definition}
Condition 1) and 2) imply that each IoT device is matched with no more than one SS unit $(n,k)$, while condition 3) and 4) show that each SS unit is matched with $q$ IoT devices for the most. If device $m$ is matched with SS unit $(n, k)$, then SS unit $(n, k)$ is also matched with device $m$, as expressed in condition 5).

\begin{figure}
  \centering
  \includegraphics[width= 3.3in]{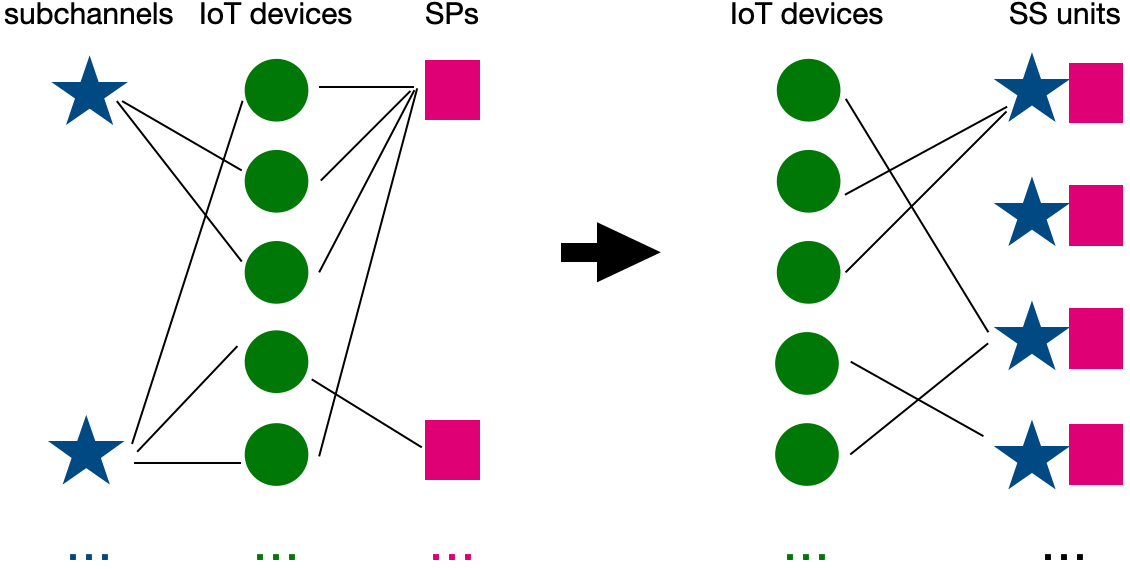}
  \caption{Illustration of converting the three-sided matching to a two-sided matching model.}
  \label{fig:matching_model}
\end{figure}

From Definition \ref{def:three_sided_matching}, we can observe that the formulated three-sided matching model between IoT devices, subchannels and SPs is equivalent to a two-sided many-to-one matching model between IoT devices and SS units, as illustrated in Fig.~\ref{fig:matching_model}. The \textit{social welfare} of the proposed matching model is defined as the network overall EE given in (\ref{eq:sum_ee}). In traditional matching problems, all the players chase for their own large profits. Different from that, our proposed matching game is a centralized one, where the central controller takes both the social welfare and selfish nature of players into consideration. To better describe the competition behavior and decision process of each agent, we assume that $m$, $k$, $n$ have the preference lists over $\mathcal{N} \times \mathcal{K}$, $\mathcal{M}\times \mathcal{N}$, and $\mathcal{M}\times \mathcal{K}$, respectively. The preference list of each agent is actually set by the central controller according to the social welfare, i.e., network EE. For example, the preference of IoT device $m$ over SS units $\mathcal{N} \times \mathcal{K}$ can be presented as:
\begin{equation}
	(n, k) \succ _{m} (n', k') \Leftrightarrow EE_{m, n, k} > EE_{m, n', k'}.
\end{equation}
From (\ref{eq:sinr_mnk}), it is worth noting that the social welfare not only depends on which SS unit that each device matches with, but also on the set of devices matched to the same SS unit, due to the existence of co-channel interference. This interdependency of players' matching status is named as ``peer effects"\cite{bodine2011peer}.
\begin{remark}
	The formulated matching game is lack of the property of substitutability.
\end{remark}
\begin{proof}
	Faced with a set of IoT devices $\mathcal{M}$, a SS unit $(n, k)$ can determine which subset of $\mathcal{M}$ it would most prefer to match with. We call this the choice of SS unit $(n, k)$ among IoT devices, denoted by $Ch_{(n, k)}(\mathcal{M}) = \mathcal{M}^*$. In other words, for any subset $\mathcal{M}'$ of $\mathcal{M}$, SS unit $(n, k)$'s most preferred subset $\mathcal{M}^*$ satisfies $\mathcal{M}^*\succ_{(n, k)} \mathcal{M}'$. 
	
	A SS unit $(n, k)$'s preferences over sets of IoT devices have the property of \textit{substitutability} if, for any set $\mathcal{M}$ that contains IoT device $m$ and $m'$, if $m$ is in $Ch_{(n, k)}(\mathcal{M})$, then $m$ is in $Ch_{(n, k)}(\mathcal{M}\setminus \{m'\})$. In the formulated matching game model, due to the existence of co-channel interference, the achievable EE over SS unit $(n, k)$ with IoT device $m$ may change after $m'$ is unmatched with $(n, k)$, and thus, $m$ may not be in $\mathcal{M}^*$ anymore. This concludes that the formulated matching game does not have the property of substitutability.  
\end{proof}

\begin{remark}
The formulated matching model has peer effects. That is, the IoT devices care not only where they are matched, but also which other devices are matched to the same place.
\end{remark}

Due to the existence of peer effects, the preference lists of IoT devices change with the matching game proceeds, which makes the matching process different from the ones with fixed preference lists \cite{7115904, 6848847}. To better describe the matching process of players with the existence of peer effect, we first introduce the concept of \textit{swap matching} as follows:

\begin{definition}
Given a matching $\Omega$ with $\Omega(m) = (n, k)$ and $\Omega(m') = (n', k')$, the \emph{swap matching} between IoT device $m$ and $m'$ is given by
\begin{align}
\Omega_{m}^{m'} = & \{\Omega\backslash \{(m, (n, k)), (m', (n', k'))\}\} \cup \nonumber\\
& \{(m, (n', k')), (m', (n, k))\} \nonumber.
\end{align}
\end{definition}
Specifically, a swap matching enables device $m$ and $m'$ switch places while keeping other devices and SS units' matchings unchanged. It is worth noticing that one of the IoT devices involved in the swap can be a “hole” representing an open spot, thus allowing a device moving to available vacancies of SS units. Similarly, one of the SS units involved in the swap can be a “hole” so that unmatched devices can be involved in the matching process.

Based on the concept of swap matching, the \textit{swap-blocking pair} is defined as the following:
\begin{definition}
Given a matching $\Omega$ with $\Omega(m) = (n, k)$ and $\Omega(m') = (n', k')$, $(m, m')$ is a \emph{swap-blocking pair} if and only if

1) $EE(\Omega_{m}^{m'}) > EE(\Omega)$;

2) $\sqrt{(x_m - x_{k'})^2 + (y_m - y_{k'})^2} \leq H \cot \theta^{thr}$;

3) $\sqrt{(x_{m'} - x_{k})^2 + (y_{m'} - y_{k})^2} \leq H \cot \theta^{thr}$.

\end{definition}

The above definition implies that if two IoT devices want to switch between two SS units, the central controller must approve the swap with the given conditions satisfied. Condition~1) implies that the network EE should be improved after the swap operation. Condition 2) and 3) indicate that the IoT devices should be within the UAV's LoS coverage after the swap. 

As stated in \cite{roth1992two}, there is no longer a guarantee that a traditional “pairwise-stability” exists when players care about more than their own matching, and, if a stable matching does exist, it can be computationally difficult to find. The authors in \cite{bodine2011peer} focused on the \textit{exchange-stable} matchings, which is defined as follows:
\begin{definition}
A matching $\Omega$ is \emph{exchange-stable} if there does not exist a swap-blocking pair.
\end{definition}

The exchange stability is a distinct notion of stability compared to the traditional notion of stability of \cite{roth1992two}, but one that is relevant to our situation where agents can compare notes with each other.

%
%
%
\subsection{Proposed Joint Device Scheduling and Subchannel Allocation Algorithm}
To find a exchange-stable matching between IoT devices and SS units, we propose two joint device scheduling and subchannel allocation algorithms, i.e., JDSSA-1 and JDSSA-2, based on multiple swap operations. For both JDSSA-1 and JDSSA-2, the algorithm consists of two main phases: the initialization phase and swap-matching phase. The initialization phase is to get a initial matching status between devices and SS units. One easy option is to randomly match devices and SS units, which however will lead to more iterations in the swap-matching phase and reduce the algorithm efficiency. The better approach is to adopt traditional deferred acceptance (DA) algorithm\cite{12} which relies on the preferences of players on the other side of the players. Since the social welfare is defined as the network EE and the power allocation is assumed to be fixed at initialization phase, it is straightforward to give the preference values of IoT devices on SS units as the achievable channel power gain. However, different from traditional matchings, the IoT devices' preferences are not \textit{strict}, which means that one IoT device may be indifferent between two SS units. This is caused by the fact that the channels from IoT devices to the UAV are dominated by LoS link and the same on different subchannels, as shown in (\ref{eq:a2g_channel}). In this case, traditional DA algorithm is not suitable to obtain the initial matching in our model. 

To well handle the non-strict preferences, we propose the Initialization Algorithm (IA) which is inspired by the traditional DA algorithm but with some changes to fit our scenario. The basic idea is to set a preference list for each IoT device over the SPs. Each device proposes to the SP with highest channel power gain and then occupies a subchannel in a round robin way to guarantee even spread of spectrum resources (lines 9-19). If the number of matched IoT devices with the SS unit is smaller than quota, the proposal is accepted (lines 10-12); otherwise, the proposal is rejected. If the proposal of IoT device $m$ is rejected by SS unit $(n, k)$, and all subchannels over SP $k$ get saturated with quota $q$, $k$ is removed from the preference list of $m$ (lines 18-19). This process continues until all the IoT devices are matched or the preference lists of all the unmatched IoT devices get empty.
The details of IA can be found in \textbf{Algorithm 1}.

\begin{algorithm}[p]
\caption{Initialization Algorithm (IA)}
\label{alg:IA}
\LinesNumbered
\KwIn{Set of IoT devices $\mathcal{M}$, set of subchannels $\mathcal{N}$, set of SPs $\mathcal{K}$, set of unmatched IoT devices $\mathcal{M}_{unmatched} = \mathcal{M}$}
\KwOut{Initial matching $\Omega$}
Set preference lists of IoT devices over SPs which satisfy the LoS range constraint, i.e., $\mathcal{P} _m, \forall m$, in descending order with respect to channel power gain according to (\ref{eq:a2g_channel})\;
Set a pointer as the indicator pointing at the first subchannel in $\mathcal{N}$ over each SP, i.e., $pointer_k = 1, \forall k$\;
Set the number of currently associated IoT devices of each SS unit to 0, i.e., $c_{n, k} = 0$, $\forall n, k$\;
\While{$\mathcal{M}_{unmatched} \neq \emptyset$}{
\For{$ m \in \mathcal{M}_{unmatched}$}{
		\If{$\mathcal{P}_m = \emptyset$}{
			Remove $m$ from $\mathcal{M}_{unmatched}$\;
			\textbf{continue}\;
		}
		IoT device $m$ proposes to its currently most preferred SP $k$\;
		\If{$c_{pointer_k, k} < q$}{
			Match device $m$ with SS unit $(pointer_k, k)$\;
			Remove $m$ from $\mathcal{M}_{unmatched}$\;	
		}
		\If{$pointer_k < N$}{
			$pointer_k ++$\;
		}
		\Else{
			\If{$c_{pointer_k, k} < q$}{
				$pointer_k = 1$\;
			}
			\Else{
				Remove $k$ from $\mathcal{P}_m$\;
			}
		}
}
}
\textbf{return} $\Omega$
\end{algorithm}

\begin{algorithm}
\caption{Joint Device Scheduling and Subchannel Allocation Algorithm (JDSSA-1)}
\LinesNumbered
\KwIn{Set of IoT devices $\mathcal{M}$, set of SS units $\mathcal{N} \times \mathcal{K}$}
\KwOut{Final matching $\Omega ^*$}
----------\textbf{Initialization}\;
Obtain the initial matching $\Omega$ with \textbf{IA}\;
----------\textbf{Swap matching process}\;
\Repeat{there exists no more swap-blocking pair}{
	\For{$m \in \mathcal{M}$}{
		\For{$m' \in \mathcal{M} \setminus \{m\}$}{
			\If{$(m, m')$ forms a blocking pair}{
				$\Omega = \Omega _m^{m'}$\;
			}
		}
	}
}
\textbf{return} $\Omega ^*$
\label{alg:JDSSA-1}
\end{algorithm}

After the initialization of matching state based on IA, swap operations among IoT devices are enabled to further improve the performance of the joint device scheduling and subchannel allocation algorithm (JDSSA-1). The details of JDSSA-1 are shown in \textbf{Algorithm 2}, where the central controller keeps searching for swap-blocking pairs to execute the swap operations (lines 4-9). The swap matching process continues until there is no more swap-blocking pairs, which reaches to the end of the algorithm (lines 9-10).
\subsection{Advanced Matching Algorithm with Exploration}
Note that the swap operations among IoT devices are based on the initial matching status, which means the final matching is highly affected by the initial matching. Since the IoT devices can swap only between their current matching states, a better matching state that can achieve higher network EE may not be formed directly based on the current matching state. This can be explained more explicitly by the following example. If the current matching state is $\{\{m, n, k\}, \{m', n', k'\}, \{m'', n'', k''\}\}$ and the optimal matching\footnote[1]{The optimal matching here is defined as the matching that can achieve the highest network EE.} is $\{\{m, n', k'\}, \{m', n'', k''\}, \{m'', n, k\}\}$, the optimal matching can not be reached if $(m, m')$ is not a swap-blocking pair under the current matching state. 

To solve this issue and further improve the performance of JDSSA-1, we introduce the concept of \textit{exploration} \cite{arnold2002dynamic} to explore the space of matching states. The basic idea is to enable players to destabilize a state involving a dominated allocation, at the cost of a temporary loss in utility. We propose the algorithm JDSSA-2, as shown in \textbf{Algorithm 3}, to involve exploration during swap operations among IoT devices. In JDSSA-2, the initialization phase is the same as that in JDSSA-1. During the swap matching phase, if a pair of IoT devices $(m, m')$ forms a swap-blocking pair, the swap operation between $m$ and $m'$ happens with probability $1$ (lines 7-8). Otherwise, the swap operation between $m$ and $m'$ happens with probability $\epsilon$ through exploration (lines 10-12). Note that $0 < \epsilon \ll 1$ is a small number that corresponds to the probability that a player makes an irrational decision.

\begin{algorithm}
\caption{Joint Device Scheduling and Subchannel Allocation Algorithm (JDSSA-2)}
\LinesNumbered
\KwIn{Set of IoT devices $\mathcal{M}$, set of SS units $\mathcal{N} \times \mathcal{K}$, maximum number of iterations $t_{max}$}
\KwOut{Final matching $\Omega ^*$}
----------\textbf{Initialization}\;
Obtain the initial matching $\Omega$ with \textbf{IA}\;
----------\textbf{Swap matching process with exploration}\;
\While{$t \leq t_{max}$}{
	\For{$m \in \mathcal{M}$}{
		\For{$m' \in \mathcal{M}\setminus \{m\}$}{
			\If{$(m, m')$ forms a blocking pair}{
				$\Omega = \Omega _m^{m'}$\;
			}
			\Else{
				Generate a random number $rand$ between $0$ and $1$\;
				\If{$rand < \epsilon$}{
					$\Omega = \Omega _{m}^{m'}$
				}
			}
		}
	}
	$t ++$\;
}
\textbf{return} $\Omega ^*$
\label{alg:JDSSA-2}
\end{algorithm}

\subsection{Property Analysis}
Given the proposed algorithms above, we present some important remarks on the properties in terms of stability, convergence, complexity and optimality.
\subsubsection{Stability}
To better capture the characteristics of the formulated matching problem, we introduce the exchange stability as stated in \textbf{Definition 5}.
\begin{lemma}
	The final matching $\Omega^*$ of JDSSA-1 is an exchange-stable matching between IoT devices and SS units.
\end{lemma}
\begin{proof}
	Assume that there exists a swap-blocking pair $(m, m')$ in the final matching $\Omega^*$ satisfying that $EE({\Omega^*}_{m}^{m'}) > EE(\Omega^*)$. According to JDSSA-1, the algorithm does not terminate until all the swap-blocking pairs are eliminated. In other words, $\Omega^*$ is not the final matching, which causes conflict. Therefore, there does not exist a swap-blocking pair in the final matching, and thus we can conclude that the proposed algorithm reaches to the two-sided exchange stability in the end of the algorithm.
\end{proof}

\subsubsection{Convergence}
We now prove the convergence of JDSSA-1 while  the convergence of JDSSA-2 is usually not considered as it is constrained by the maximum number of iterations $t_{max}$.
\begin{theorem}
	JDSSA-1 converges to an exchange-stable matching $\Omega^*$ within limited number of iterations.
\end{theorem}
\begin{proof}
As shown in \textbf{Algorithm 2}, the convergence of JDSSA-1 depends mainly on the \textit{swap matching process}. According to \textbf{Definition 4}, after each swap operation between IoT devices $m$ and $m'$ along with their corresponding matched SS units $(n, k)$, $(n', k')$, the social welfare satisfies: $EE({\Omega^*}_{m}^{m'}) > EE(\Omega^*)$. Note that the network EE has an upper bound due to the limited spectrum and power resources. Therefore, the number of swap operations in JDSSA-1 is limited, as such, there exists a swap operation after which no swap-blocking pair can further improve the network EE. JDSSA-1 then converges to the final matching $\Omega^*$ which is stable as proved in \textbf{Lemma 1}.	
\end{proof}

\subsubsection{Complexity}
Given the convergence of JDSSA-1, we can now compute the complexity. The computational complexity of JDSSA-1 is composed of two main parts, i.e., the IA and the swap matching process. For the IA, the complexity for IoT devices proposing to their most preferred SPs is given by $\mathcal{O}(MK)$. For the swap matching process, the complexity depends on the number of iterations (line 4) and number of swap operations (lines 5-9) within each iteration. However, the number of iterations cannot be given in a closed form. This is because it is uncertain that at which step the algorithm converges to an exchange stable matching, which is a common problem in most heuristic algorithms. We will analyze the number of total iterations for different numbers of IoT devices and subchannels in Fig.~\ref{fig:swapCdf} and Fig.~\ref{fig:num_swap_operation}, and more detailed analysis can be found in Section VII. Here, we give an approximation of the complexity with number of iterations $T$.
\begin{theorem}
With the number of iterations $T$, the computational complexity of JDSSA-1 is given by $\mathcal{O}(TM^2)$.	
\end{theorem}
\begin{proof}
	Within each iteration, the central controller searches for swap-blocking pairs and the IoT devices execute all approved swap operations over their currently matched SS units. Since there are $M$ IoT devices, at most $M^2$ swap operations need to be considered in each iteration. Therefore, given the number of iterations $T$, the computational complexity of JDSSA-1 can be presented by $\mathcal{O}(TM^2)$.
\end{proof}
The complexity of JDSSA-2 is dependent on the maximum number of iterations $t_{max}$, while the complexity of exhaustive search increases exponentially with $M$, $N$ and $K$.
\subsubsection{Optimality}
We show below the insights on the relationships between optimal and stable solutions of JDSSA-1 and JDSSA-2.
\begin{proposition}
	In JDSSA-1, all local maxima of the network EE corresponds to an exchange stable matching state, and vice versa.
\end{proposition}
\begin{proof}
	Assume that the achievable EE under matching state $\Omega$ is a local maximal value. If $\Omega$ is not exchange stable, it indicates that there exists a swap-blocking pair that can further improve the EE. However, this is inconsistent with the assumption that the EE under matching state $\Omega$ is local optimal, and hence we conclude that $\Omega$ is exchange stable. Similarly, assume that a matching state $\Omega'$ is exchange stable, which implies that the EE can not be further improved under current matching state. This concludes the local optimality of the EE under $\Omega'$.
\end{proof}
\begin{proposition}
	With sufficiently large $t_{max}$, the final matching of JDSSA-2 is exchange stable and global optimal.
\end{proposition}
\begin{proof}
	In JDSSA-2, exploration is enabled to destabilize the dominated states and search for all possible matching states. Assume that the final matching of JDSSA-2 is not global optimal, which means that there exists an optimal matching state unexplored. By increasing the number of iterations, i.e., $t_{max}$, the set of explored states expands. Since the number of combinations between IoT devices and SS units is limited, with sufficiently large $t_{max}$, the optimal matching state can finally be obtained. Under the optimal matching state, there is no approved swap operations to further improve EE. Therefore, the optimal state is also exchange stable. 
\end{proof}

\section{Power Allocation}
\label{sec:power_allocation}
\begin{algorithm}
\caption{Dinkelbach's Algorithm Based Power Allocation (DABPA)}
\LinesNumbered
\KwIn{$\tau = 0$, $\mu = 10^{(-8)}$, $f_{max} = \mu + 0.01$}
\KwOut{$p_{m}, m \in \{1, ..., \left|\mathcal{M}_{n, k}\right|\}$}
\While{$f_{max} > \mu$}{
	$\{p_m ^*\} = \arg \max \log _2\left(1 + \frac{\sum _{m = 1}^{\left|\mathcal{M}_{n,k}\right|}p_mg_m}{\sigma ^2}\right) - \tau \sum _{m = 1}^{\left|\mathcal{M}_{n, k}\right|}\left(p_m + p_0 \right)$, s.t. (\ref{eq:power_allocation_constraint1}), (\ref{eq:power_allocation_constraint2});
	
	$f_{max} = \log _2\left(1 + \frac{\sum _{m = 1}^{\left|\mathcal{M}_{n, k}\right|}p_m^*g_m}{\sigma ^2}\right) - \tau \sum _{m = 1}^{\left|\mathcal{M}_{n, k}\right|}\left(p_m^* + p_0 \right)$;
	
	$\tau = \frac{\log _2\left(1 + \frac{\sum _{m = 1}^{\left|\mathcal{M}_{n, k}\right|}p_m^*g_m}{\sigma ^2}\right)}{\sum _{m = 1}^{\left|\mathcal{M}_{n, k}\right|}\left(p_m^* + p_0 \right)}$;
}
\textbf{return} $\Omega ^*$
\end{algorithm}  

In Section~\ref{sec:subchannel_allocation}, the joint device scheduling and subchannel allocation is performed under the assumption that the power allocation scheme is known. In this section, we discuss how we allocate transmit power to IoT devices under given device scheduling and subchannel allocation stragtegy. For the given device scheduling and subchannel allocation strategy $\{\{a_{m, n}\}, \{b_{m,k}\}\}$, power allocation can be performed independently among SS units as they are with orthogonal time/frequency resources. Therefore, we can consider the EE maximization problem on each SS unit separately. Since the form of power allocation sub-problems on different SS units are the same, for simplicity, we drop the SP index $k$ and subchannel index $n$. As such, $p_{m, n, k}$ and $\gamma_{m, n, k}$ are replaced by $p_{m}$, and $\gamma _{m}$, respectively. Accordingly, the sub-problem on SS unit $(n, k)$ is expressed as
\begin{subequations}
\label{eq:power_allocation_subproblem}
\begin{equation}
	\label{eq:power_allocation_objective_function}
	\max _{\{p_{m}\}} \frac{\sum _{m = 1} ^{ \left|\mathcal{M}_{n, k}\right|}\log _2(1 + \gamma _{m})}{\sum _{m = 1} ^{\left|\mathcal{M}_{n, k}\right|}(p_{m} + P_0)},
\end{equation} 
\begin{equation}
\label{eq:power_allocation_constraint1}
s.t. \ \ \ p_m \leq P_{max}, \forall m \in \{1, ..., \left|\mathcal{M}_{n, k}\right|\},
\end{equation}
\begin{equation}
	p_m \geq P_{min}, \forall m \in \{1, ..., \left|\mathcal{M}_{n, k}\right|\}.
\label{eq:power_allocation_constraint2}
\end{equation}
\end{subequations}

The numerator in (\ref{eq:power_allocation_objective_function}) can be rewritten as
\begin{align}
\label{eq: numerator_reform}	
	& \sum\nolimits _{m = 1} ^{\left|\mathcal{M}_{n, k}\right|} \log _2 \left(1 + \gamma _{m}\right) = \sum\nolimits _{m = 1} ^{\left|\mathcal{M}_{n, k}\right|} \log _2 \left(1 + \frac{p_{m}g_{m}}{\sum _{i = m + 1}^{\left|\mathcal{M}_{n, k}\right|}p_{i}g_{i} + \sigma ^2} \right)\nonumber \\
	& = \log _2\left[\frac{\sum _{i = 1}^{\left|\mathcal{M}_{n, k}\right|} p_i g_i + \sigma^2}{\sum _{i = 2}^{\left|\mathcal{M}_{n, k}\right|} p_i g_i + \sigma^2}\cdot \frac{\sum _{i = 2}^{\left|\mathcal{M}_{n, k}\right|} p_i g_i + \sigma^2}{\sum _{i = 3}^{\left|\mathcal{M}_{n, k}\right|} p_i g_i + \sigma^2} \cdots \frac{\sum _{i = M-1}^{\left|\mathcal{M}_{n, k}\right|} p_i g_i + \sigma^2}{p_M g_M + \sigma^2} \cdot \frac{p_Mg_M + \sigma^2}{\sigma^2} \right] \nonumber \\
	& =\log _2 \left( 1 + \frac{\sum _{m = 1}^{\left|\mathcal{M}_{n, k}\right|}p_{m}g_{m}}{\sigma ^2} \right).
\end{align}
Therefore, problem (\ref{eq:power_allocation_subproblem}) can be rewritten as
\begin{subequations}
\label{eq:power_allocation_reform}
	\begin{equation}
		\max_{\{p_m\}}\frac{\log _2\left(1 + \frac{\sum _{m = 1}^{\left|\mathcal{M}_{n, k}\right|}p_mg_m}{\sigma ^2}\right)}{\sum\nolimits _{m = 1}^{\left|\mathcal{M}_{n, k}\right|}\left(p_m + p_0 \right)},
		\label{eq:power_allocation_objective_reform}
	\end{equation}
	\begin{equation}
		s.t. \ \ \ (\ref{eq:power_allocation_constraint1}), (\ref{eq:power_allocation_constraint2}). 
	\end{equation}
\end{subequations}
Define a function $f(\tau)$ as
\begin{align}
	f(\tau, \{p_m\}) = 
		& \log _2\left(1 + \frac{\sum _{m = 1}^{\left|\mathcal{M}_{n, k}\right|}p_mg_m}{\sigma ^2}\right) - \tau \sum _{m = 1}^{\left|\mathcal{M}_{n, k}\right|}\left(p_m + p_0 \right) \label{eq:eta_function}.
\end{align}
It is clear to see that (\ref{eq:power_allocation_objective_reform}) is in fractional form, and $f(\tau, \{p_m\})$ is strictly concave with respect to $\{p_m\}$ with given $\tau$. Moreover, constraint (\ref{eq:power_allocation_constraint1}) and (\ref{eq:power_allocation_constraint2}) are affine. Therefore, (\ref{eq:power_allocation_reform}) can be optimally solved by the Dinkelbach's algorithm, as shown in \textbf{Algorithm 4}. In the algorithm, the value of $\tau$ is initialized as 0. Given the $\tau$ value, we have a set of $\{f(\tau, \{p_m\})\}$ with different $\{p_m\}$ value. We find $\{p_m^*\}$ that maximizes $\{f(\tau, \{p_m\})\}$ (line 2). Record the maximum $\{f(\tau, \{p_m\})\}$ under the current $\tau$ value as $f_{max}$ (line 3). Set $\tau$ as the value satisfying $f(\tau, \{p_m*\}) = 0$ (line 4). This procedure iterates until $f_{max}$ is smaller or equal to the threshold $\mu$. 
 
%
\section{UAV's Stop Points Deployment}
Note that the UAV does not continuously move as it must stop, serve the devices, and then update its locations. Considering the limited onboard battery, the UAV is encouraged to spend a minimum total energy on mobility so as to remain operational for a longer time. Therefore, our goal is to minimize the number of SPs and ensure that each IoT device is covered by the LoS connection with UAV at one of the SPs for the sake of better channel gain. In this section, we propose a simple but effective sub-optimal approach to deploy the UAV's SPs.

Given the UAV's height H, the radius of UAV's LoS coverage range is expressed as
\begin{equation}
	R_{LoS} = H \times \cot \theta^{thr}. 
\end{equation}
Given the radius of LoS coverage range of the UAV and that of the target area, the objective is to find the minimum number of SPs and their optimal positions to cover the target area, which can be solved by exploiting the \textit{disk covering problem}~\cite{3}. In the disk covering problem, given the radius of a disk, i.e., $R$, and the number of congruent smaller disks, i.e., $K$, the minimum required radii of the small disks, i.e., $R_{min}$, to cover the target disk can be calculated. In Fig.~\ref{fig:disk-covering}, we illustrate the covering solutions for a target disk with 3 and 4 smaller disks, respectively. Suppose the target disk is centered at $(0, 0)$, the minimum required radius and center locations of small disks are shown in Table II.   

\begin{figure}
  \centering
  \subfigure[Covering solution with 3 stop points]{
    \begin{minipage}[t]{0.43\textwidth}
     \includegraphics[width=2.0 in]{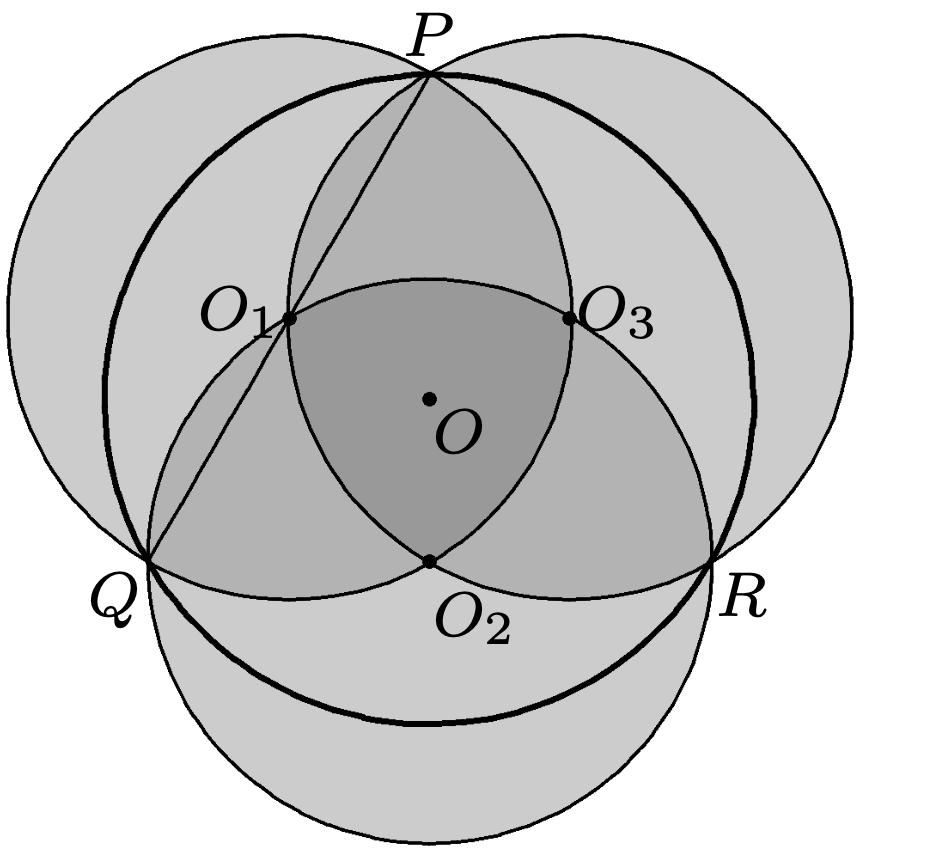}
    \end{minipage}}
  \subfigure[Covering solution with 4 stop points]{
    \begin{minipage}[t]{0.43\textwidth}
     \includegraphics[width=2.0 in]{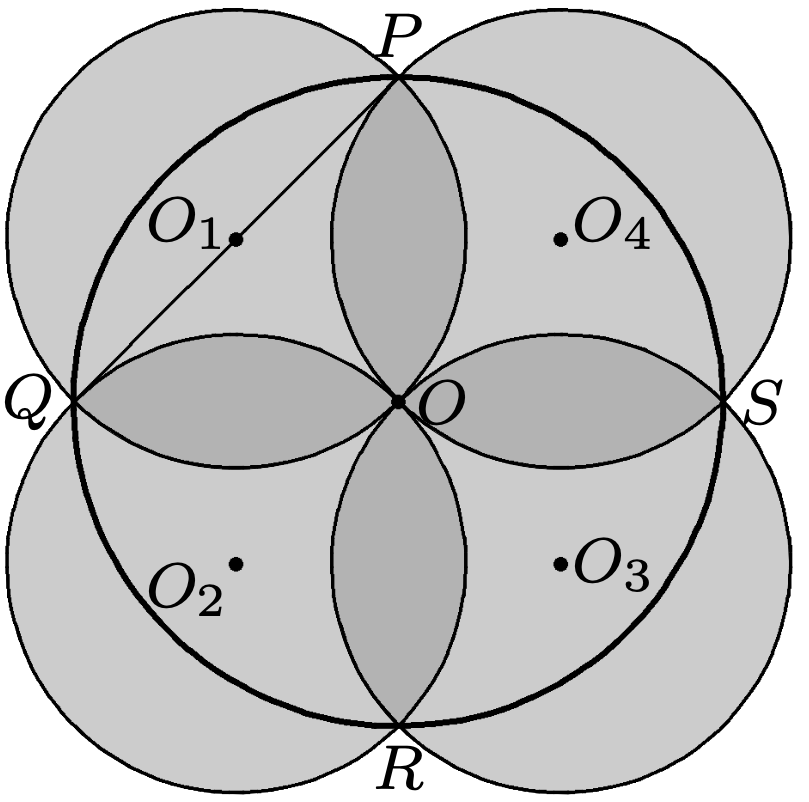}
     \end{minipage}} 
    \caption{Disk covering illustration.}
  \label{fig:disk-covering}
\end{figure}

\begin{table}[htbp]
\label{table:disk_covering}
\caption{Number, Radii, and Locations of Small Disks In Disk Covering Problem}
\begin{center}
\begin{tabular}{|l|l|l|}
\hline
$\mathbf{K}$ & $\mathbf{R_{min}}$ & \textbf{Center positions of small disks}\\
\hline
$1$, $2$ & R & $(0, 0)$\\
\hline
$3$ & $\frac{\sqrt{3}R}{2}$ & $(-\frac{\sqrt{3}R}{4}, \frac{R}{4}), (0, -\frac{R}{2}), (\frac{\sqrt{3}R}{4}, \frac{R}{4})$\\
\hline
$4$ & $\frac{\sqrt{2}}{2}R$ & $(-\frac{R}{2}, \frac{R}{2}), (-\frac{R}{2}, -\frac{R}{2}), (\frac{R}{2}, -\frac{R}{2}), (\frac{R}{2}, \frac{R}{2})$\\
\hline
$\cdots$ & $\cdots$ & $\cdots$ \\
\hline
\end{tabular}
\end{center}
\end{table}

Assume that the minimum number of SPs to cover the target area is $K^*$, then the following condition should hold:
\begin{equation}
	R_{min}(K^*) < R_{LoS} < R_{min}(K^* + 1).
\end{equation}
According to Table II, the minimum number and optimal locations of SPs can then be easily obtained given the values of $R_{LoS}$ and $R$.

\begin{figure}
  \centering
  \includegraphics[width= 3.3in]{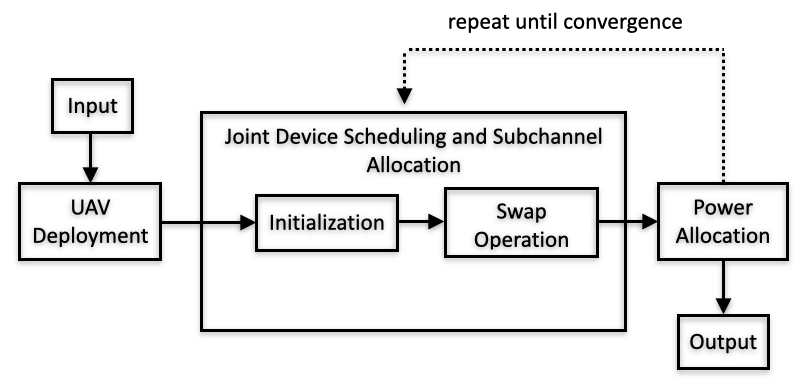}
  \caption{Flow chart of the joint UAV deployment, device scheduling, and resource allocation algorithm.}
  \label{fig:whole_algorithm}
\end{figure}

After figuring out the UAV deployment, device scheduling, subchannel allocation and power allocation strategies, it is worth considering how to solve the joint problem optimally. We propose a joint UAV deployment, device scheduling and resource allocation algorithm (JUDDSRA) as shown in Fig.~\ref{fig:whole_algorithm}. We first solve the UAV deployment problem as discussed in this section, after which the device scheduling and subchannel allocation problem is jointly solved through JDSSA-1 or JDSSA-2. As the matching status between IoT devices and SS units and the power allocation inside each SS unit jointly influence EE, the power allocation is executed through DABPA iteratively after the convergence of JDSSA-1/JDSSA2. This process repeats until the results of device scheduling, subchannel allocation and power allocation keep unchanged, which means the joint algorithm finally converges.  
\section{Numerical Results and Discussions}
\label{sec:numerical-results}
In this section, we investigate the performance of the proposed joint UAV deployment, device scheduling and resource allocation algorithm through simulations. For our simulations, the IoT devices are randomly located within a geographical area of radius $350$ m. The specific parameter value settings are summarized in Table III unless otherwise specified. When applicable, we compare our results with pre-deployed stationary UAV scenario and OMA scenario while adopting the subchannel and power allocation algorithms proposed in Section IV and V. All statistical results are averaged over a large number of independent runs. 
\begin{table}[htbp]
\caption{Simulation Parameters}
\begin{center}
\begin{tabular}{|p{0.07\textwidth}|p{0.35\textwidth}|p{0.10\textwidth}|}
\hline
\textbf{Parameter} & {\textbf{Description}} & {\textbf{Value}} \\
\hline
\hline
{$R$} & Cellular radius & {$350$ m}\\
\hline
{$H$} & UAV altitude & {$150$ m}\\
\hline
{$N$} & Total number of subchannels & {$5$}\\
\hline
{$\theta^{thr}$} & Minimum elevation angle of LoS coverage & {$\frac{\pi}{6}$} \\
\hline
{$\eta$} & Unit power gain at $d_0 = 1$ m & {$1.4\times 10^{-4}$}\\
\hline
{$\sigma^2$} & Noise power & {$-99$ dBm} \\
\hline
{$P^{max}$} & Maximum transmit power of IoT devices & {$500$ mW}\\
\hline
{$P_{min}$} & Minimum transmit power of IoT devices & {$100$ mW}\\
\hline
{$P_{0}$} & Fixed circuit power consumption & {$1$ mW} \\
\hline
{$q$} & Maximum number of IoT devices associated to the same SS unit & {$3$}\\
\hline
\end{tabular}
\end{center}
\label{table:2}
\end{table}
\subsection{Snapshot of the UAV deployment and device scheduling status}
Fig.~\ref{fig:snapshot} plots the 2-D system realization with 60 IoT devices deployed in the cell. In this snapshot, the coloured stars represent projection points of the UAV's SPs on the ground,  while dots in the same colour remark the IoT devices scheduled to transmit to the UAV at corresponding SPs. Dots in black indicate the IoT devices not scheduled for transmission. According to Table II and the values of area radius and minimum elevation angle of UAV LoS coverage stated in Table III, the UAV is supposed to stop at $4$ SPs, i.e.$(-\frac{R}{2}, \frac{R}{2}), (-\frac{R}{2}, -\frac{R}{2}), (\frac{R}{2}, -\frac{R}{2}), (\frac{R}{2}, \frac{R}{2})$, assuming the area is centered at $(0, 0)$. It is shown in Fig.~\ref{fig:snapshot} that the IoT devices are scheduled to transmit to UAV at the nearest SP for the sake of higher channel power gain. 
\begin{figure}
  \centering
  \includegraphics[width= 3.3in]{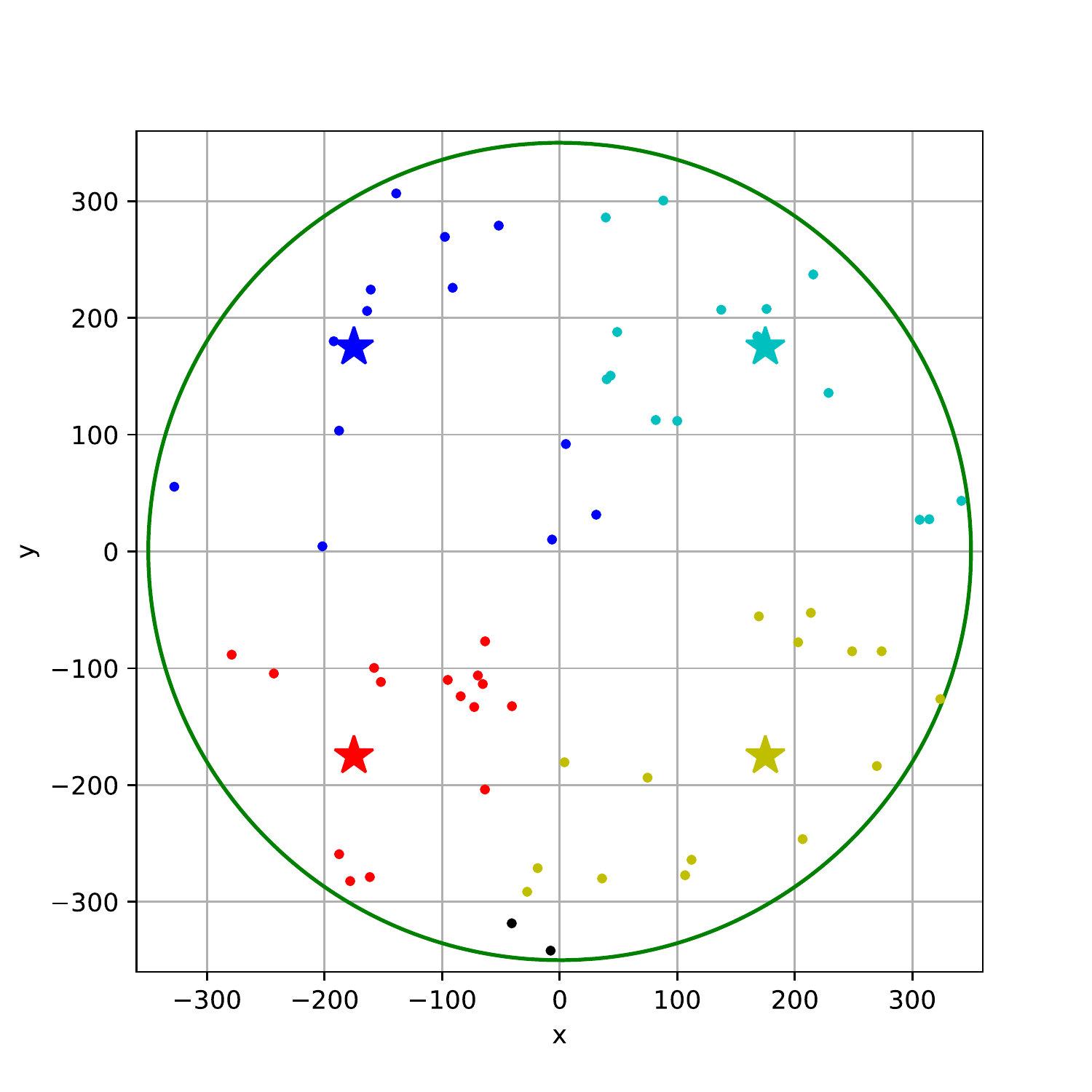}
  \caption{UAV deployment and device scheduling status with 60 IoT devices.}
  \label{fig:snapshot}
\end{figure}

\subsection{Convergence of the proposed algorithms}
Fig.~\ref{fig:swapCdf} shows the cumulative distribution function (CDF) of the number of swap operations for the matching process in JDSSA-1, and thus proves the convergence of JDSSA-1 as stated in \textbf{Theorem 1}. The CDF shows that JDSSA-1 converges within a small number of swap operations for different number of IoT devices. For example, when there are $50$ IoT devices in the network, on average a maximum of $80$ swap operations is required for JDSSA-1 to converge. It can also be observed that the number of swap operations increases with larger number of IoT devices, due to the improved probability of the existence of swap-blocking pairs.   
\begin{figure}
  \centering
  \includegraphics[width= 3.3in]{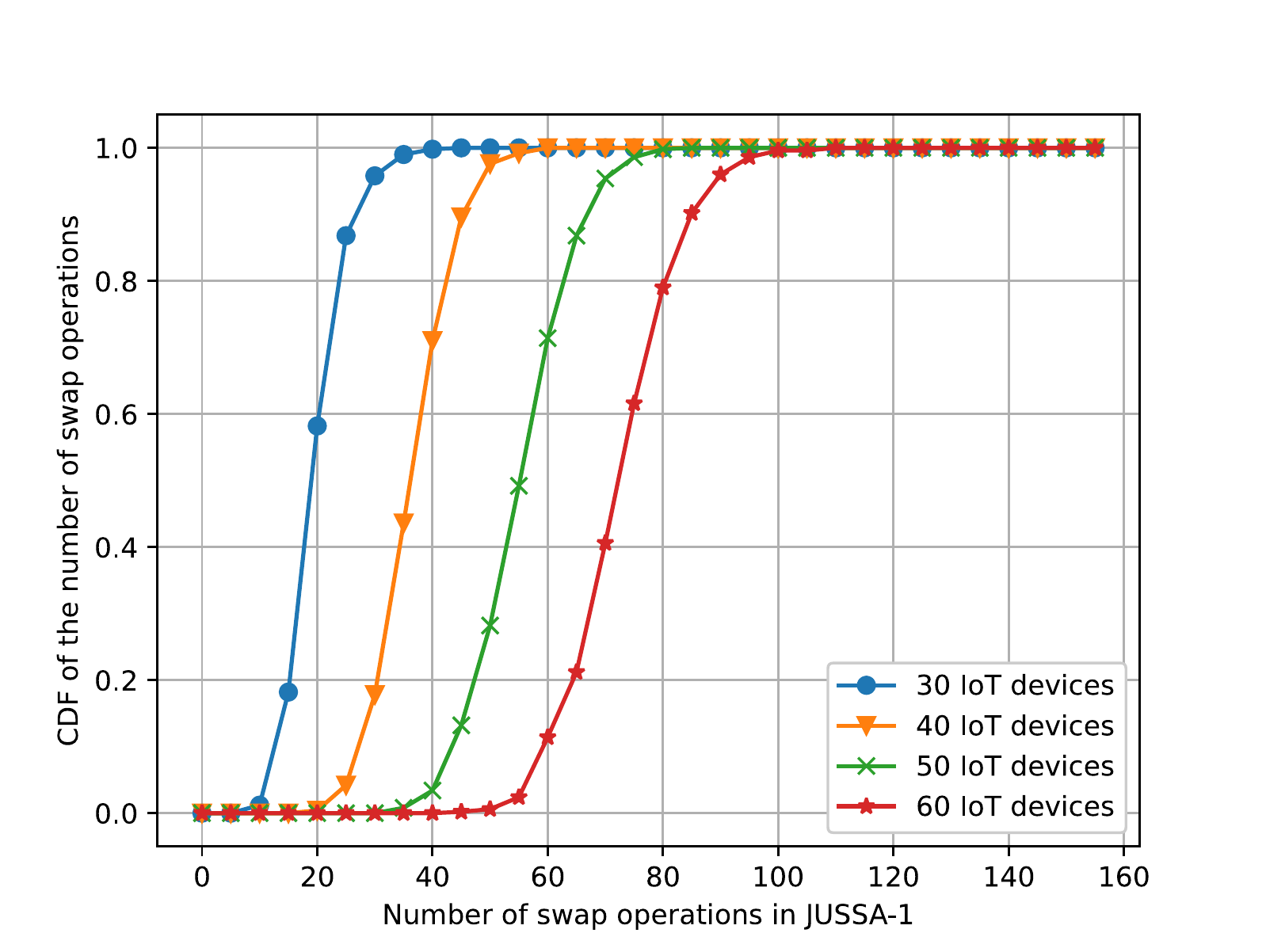}
  \caption{CDF of the number of swap operations in JDSSA-1.}
  \label{fig:swapCdf}
\end{figure}

\begin{figure}
  \centering
  \includegraphics[width= 3.3in]{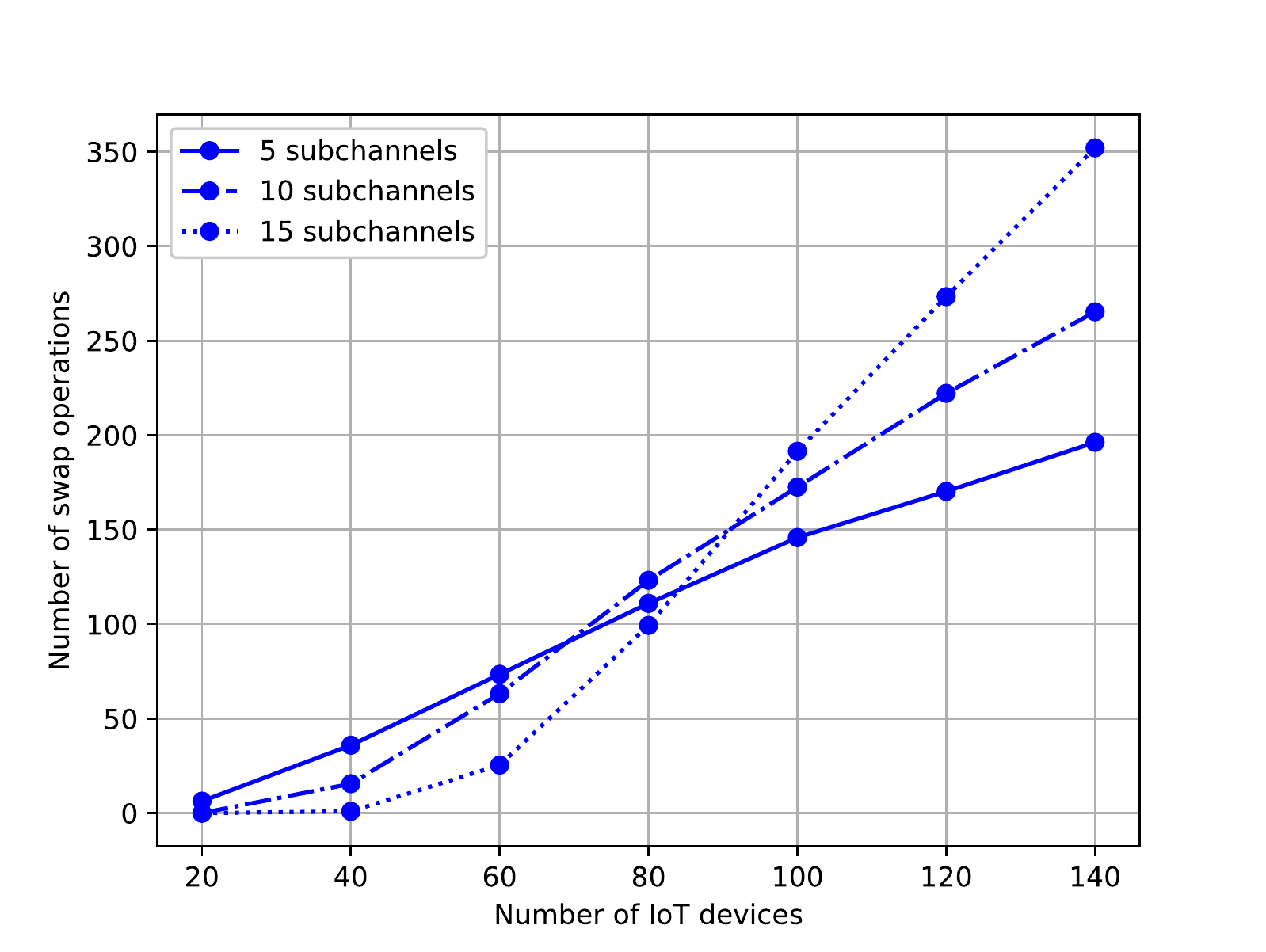}
  \caption{Number of swap operations versus number of subchannels.}
  \label{fig:num_swap_operation}
\end{figure} 

Fig.~\ref{fig:num_swap_operation} demonstrates the number of swap operations in JDSSA-1 with different number of subchannels and IoT devices. One can observe that, under small number of IoT devices, the number of swap operations decreases with the number of subchannels. This is due to the fact that, when the number of subchannels times the number of SPs is larger than the number of IoT devices, all the IoT devices occupy dedicated subchannels after IA, which means there is no need for swap operations as each IoT device has the same channel power gain over different subchannels. However, with the increment of the number of IoT devices, the number of swap operations improves faster for the cases with larger number of subchannels. This is because IoT devices have more space of paired subchannels to exchange to obtain higher EE. 
\subsection{NOMA case versus OMA case}
Results for demonstrating the comparison between UAV-aided IoT networks with NOMA and OMA protocol are shown in Fig.~\ref{fig:EE_OMA} and Fig.~\ref{fig:access_num}. As a baseline algorithm, the result of OMA case is obtained by running the proposed algorithm JUDDSRA with the maximum number of IoT devices occupying each SS unit, i.e., quota, set to be $1$. Moreover, we also present the results for the quota value set to be $2$ and $3$, respectively. From Fig.~\ref{fig:EE_OMA}, we can see that the EE first increases with the number of IoT devices for both NOMA and OMA case. Then, after a threshold, the EE for OMA case tends to be stable. This is because the number of accessed IoT devices is restricted to a certain value after which EE increases slightly due to the multi-user diversity. For both $q = 2$ and $q = 3$ case, EE begins to decrease at some certain points and finally gets stable when the number of IoT devices reaches to $50$ and $70$, respectively. This decrement can be explained by the higher interference introduced by more devices sharing the same SS unit. After  all the SS units get saturated, which is restricted by the quota value, EE changes slightly with the multi-user diversity. We can also observe from Fig.~\ref{fig:EE_device_num} that EE decreases when more IoT devices are allowed to occupy one SS unit, which is because of the higher interference caused by spectrum sharing.
\begin{figure}
  \centering
  \includegraphics[width= 3.3in]{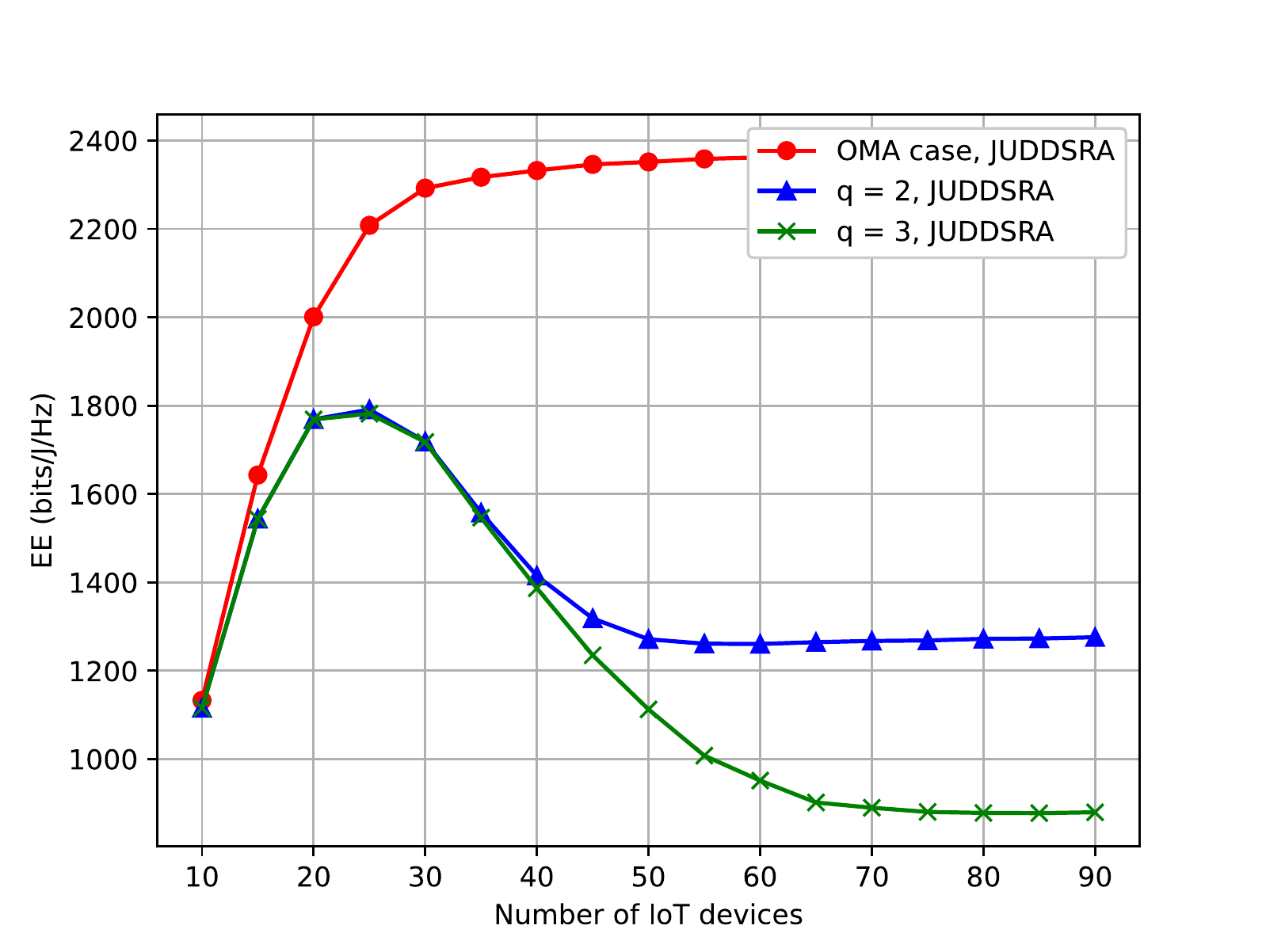}
  \caption{Comparison of average EE between NOMA and OMA cases.}
  \label{fig:EE_OMA}
\end{figure}

\begin{figure}
  \centering
  \includegraphics[width= 3.3in]{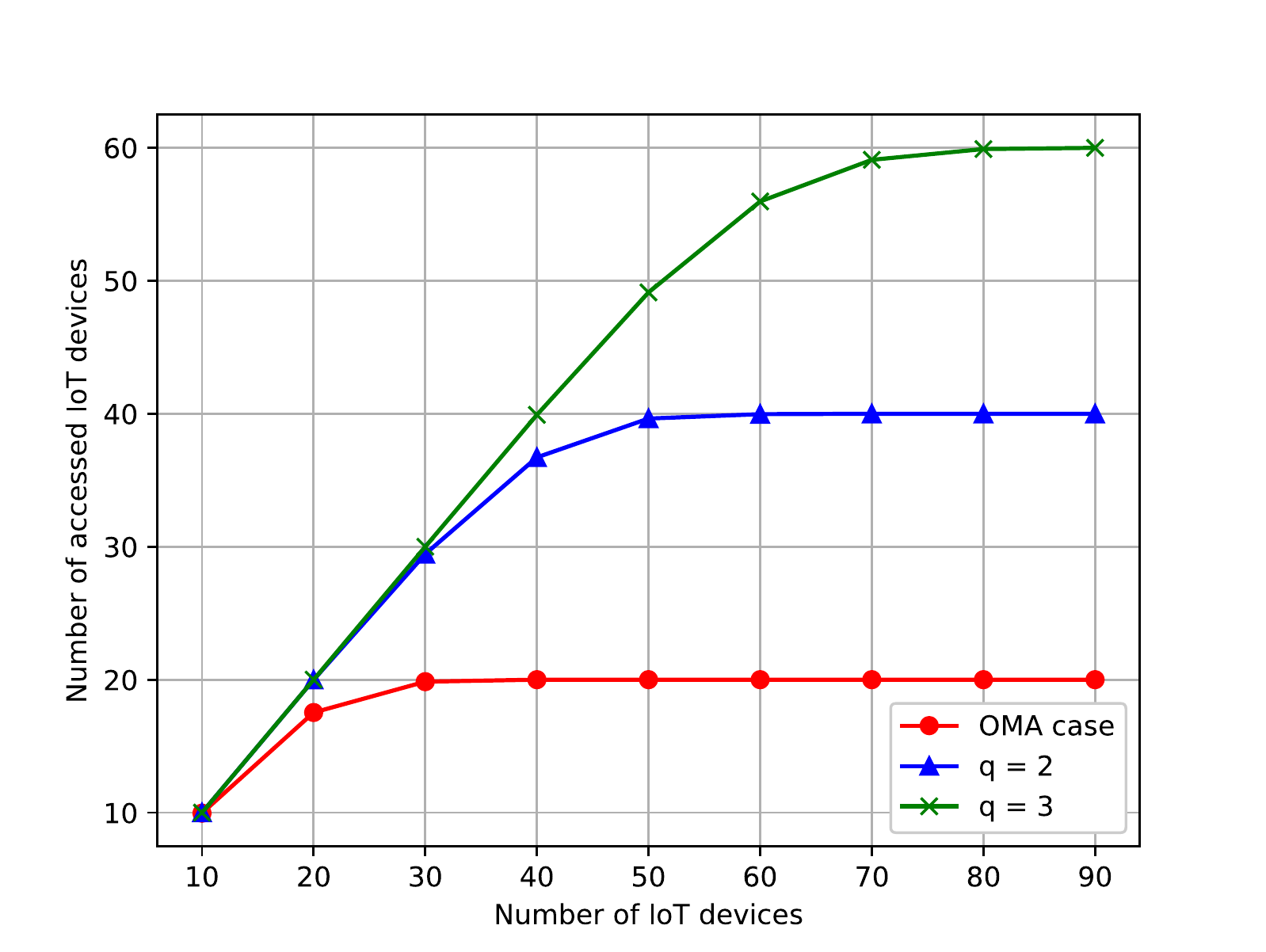}
  \caption{Number of accessed IoT devices with different number of IoT devices.}
  \label{fig:access_num}
\end{figure}

Fig.~\ref{fig:access_num} shows the number of accessed IoT devices for NOMA and OMA cases. As expected, NOMA can support more devices to upload data to UAV compared to the OMA case. In addition, the number of accessed devices increases with higher quota value. When all SS units get saturated with the predefined quota value, the number of accessed IoT devices remains unchanged even when the number of IoT devices in network keeps increasing.

\subsection{JDSSA-1 versus JDSSA-2}
Fig.~\ref{fig:EE_JDSSA2} shows the performance of JDSSA-1 and JDSSA-2 under different numbers of iterations, i.e., $t_{max}$, where we set fixed power transmission of IoT devices with $500$~mW. One can observe that, with lower number of iterations, JDSSA-2 achieves worse performance compared to JDSSA-1 because of the introduced irrational decisions. However, with the increment of iterations, the performance achieved by JDSSA-2 is improved thanks to the expanded region of matching states. For example, with $t_{max} = 10^5$, JDSSA-2 achieves obvious improvement on EE compared to JDSSA-1 when the number of IoT devices reaches to $30$.

\begin{figure}
  \centering
  \includegraphics[width= 3.3in]{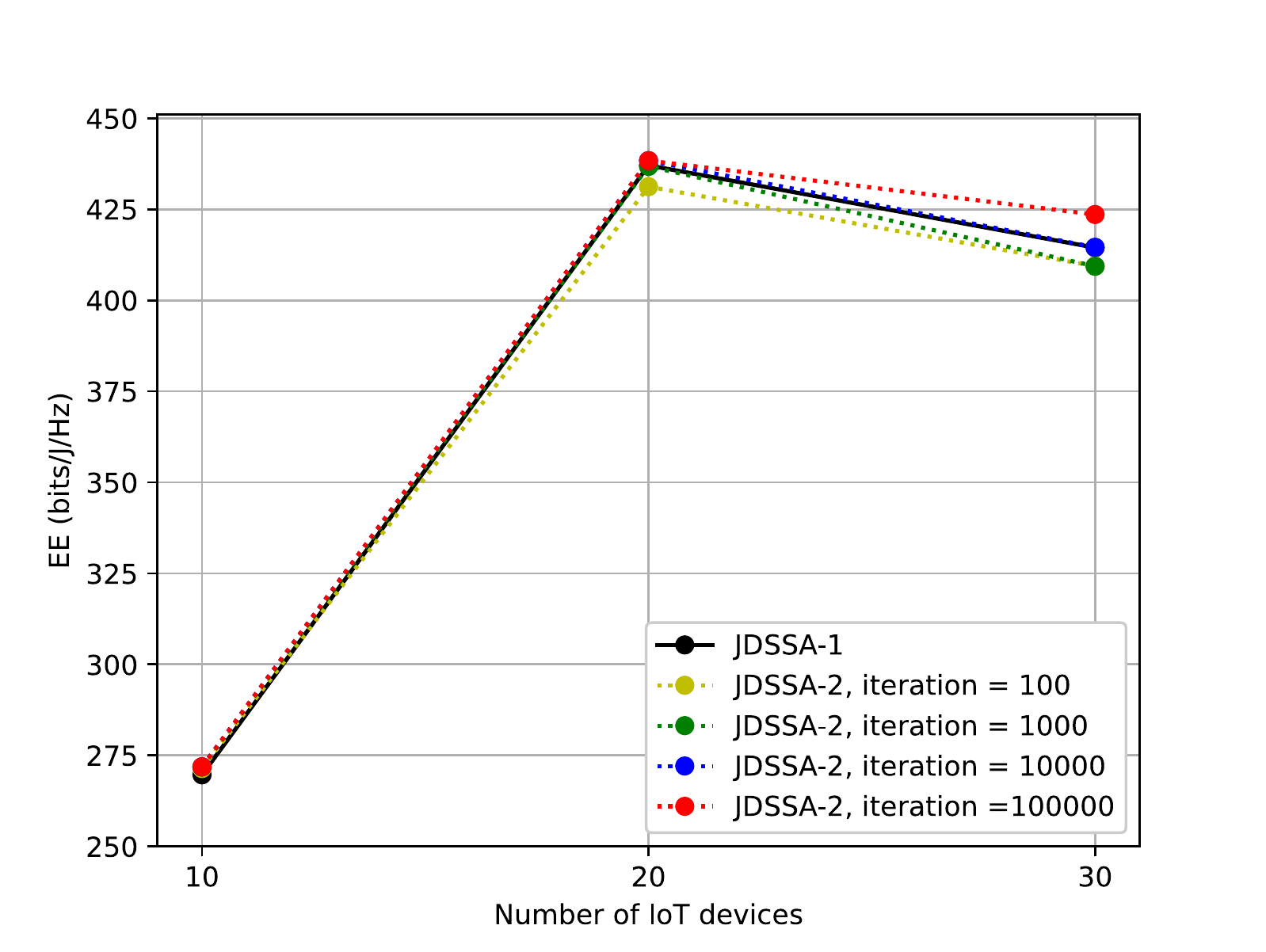}
  \caption{Comparison of average EE between JDSSA-1 and JDSSA-2 with $q = 2$.}
  \label{fig:EE_JDSSA2}
\end{figure} 

\subsection{JUDDSRA versus benchmark algorithms}
In Fig.~\ref{fig:EE_device_num}, we compare the performance of JUDDSRA with the fixed power allocation scheme. The fixed power allocation scheme is given by allowing each IoT device to transmit with the maximum power $500$ mW, while applying the proposed UAV deployment scheme and JDSSA-1 for UAV deployment and joint device scheduling and subchannel allocation, respectively. We also demonstrate the EE performance under different quota values. It can be seen that JUDDSRA achieves much higher EE than the fixed power allocation scheme, which therefore illustrates the importance of applying DABPA for power allocation.
\begin{figure}
  \centering
  \includegraphics[width= 3.3in]{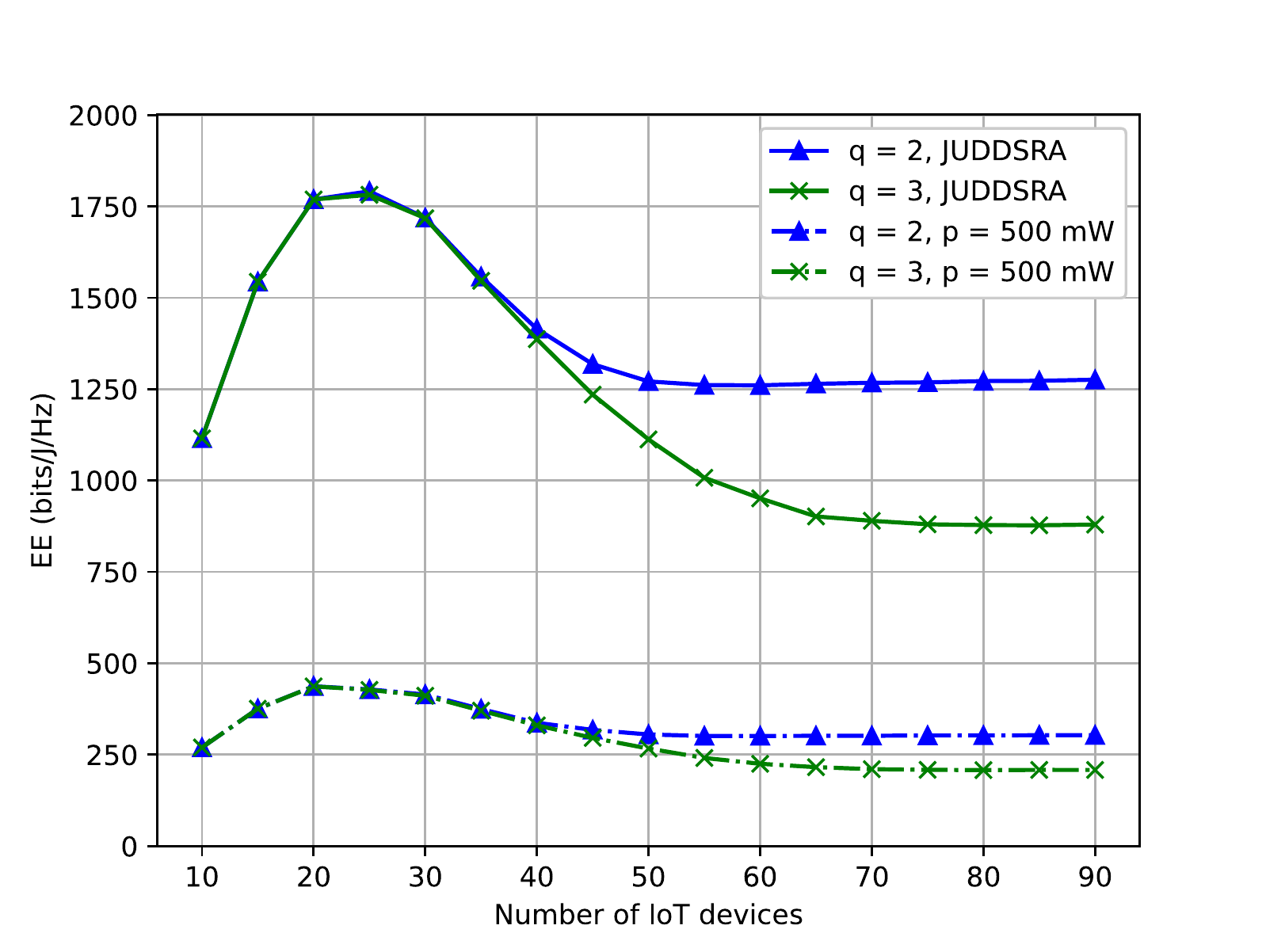}
  \caption{Comparison of average EE between JUDDSRA and fixed power allocation scheme.}
  \label{fig:EE_device_num}
\end{figure}

Fig.~\ref{fig:EE_max_power} demonstrates the comparison of average EE between JUDDSRA and matching algorithm without swap operation. For the latter one, we apply the same UAV deployment and power allocation schemes as in JUDDSRA except that IA is adopted for matching between IoT devices and SS units. For better showing the influence of the maximum transmit power on EE performance, we also expand the feasible region by lowering the minimum transmit power to $-40$~dBm. The figure implies that JUDDSRA achieves much higher EE compared to the matching algorithm without swap operation, which therefore reflects the importance of swap operation in matching problems with peer effects. It is also shown that EE firstly increases with maximum transmit power but gets saturated at some certain points due to the tradeoff between achievable data rate and power consumption.  
\begin{figure}
  \centering
  \includegraphics[width= 3.3in]{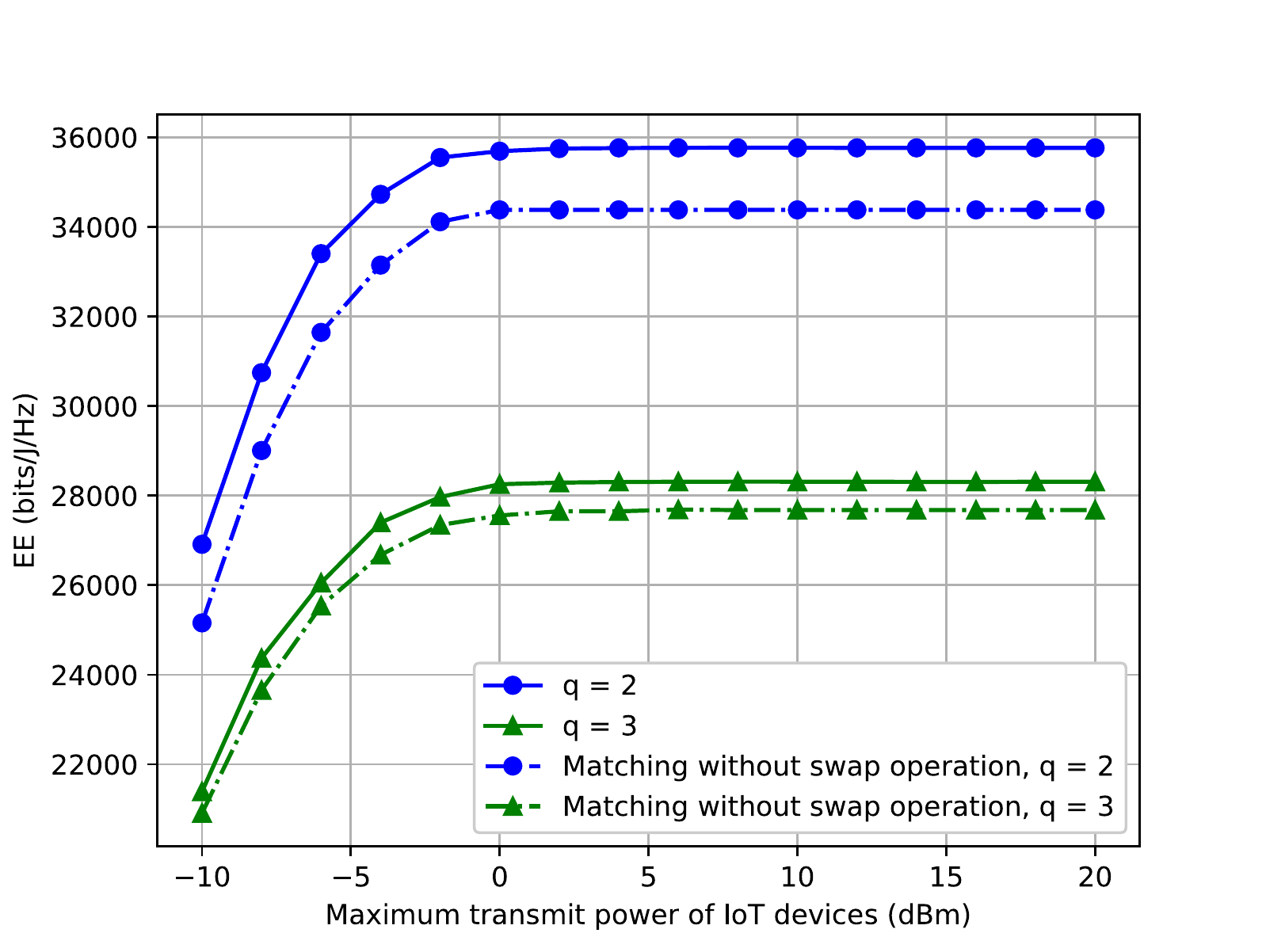}
  \caption{Comparison of average EE between JUDDSRA and matching algorithm without swap operation.}
  \label{fig:EE_max_power}
\end{figure}

\begin{figure}
  \centering
  \includegraphics[width= 3.3in]{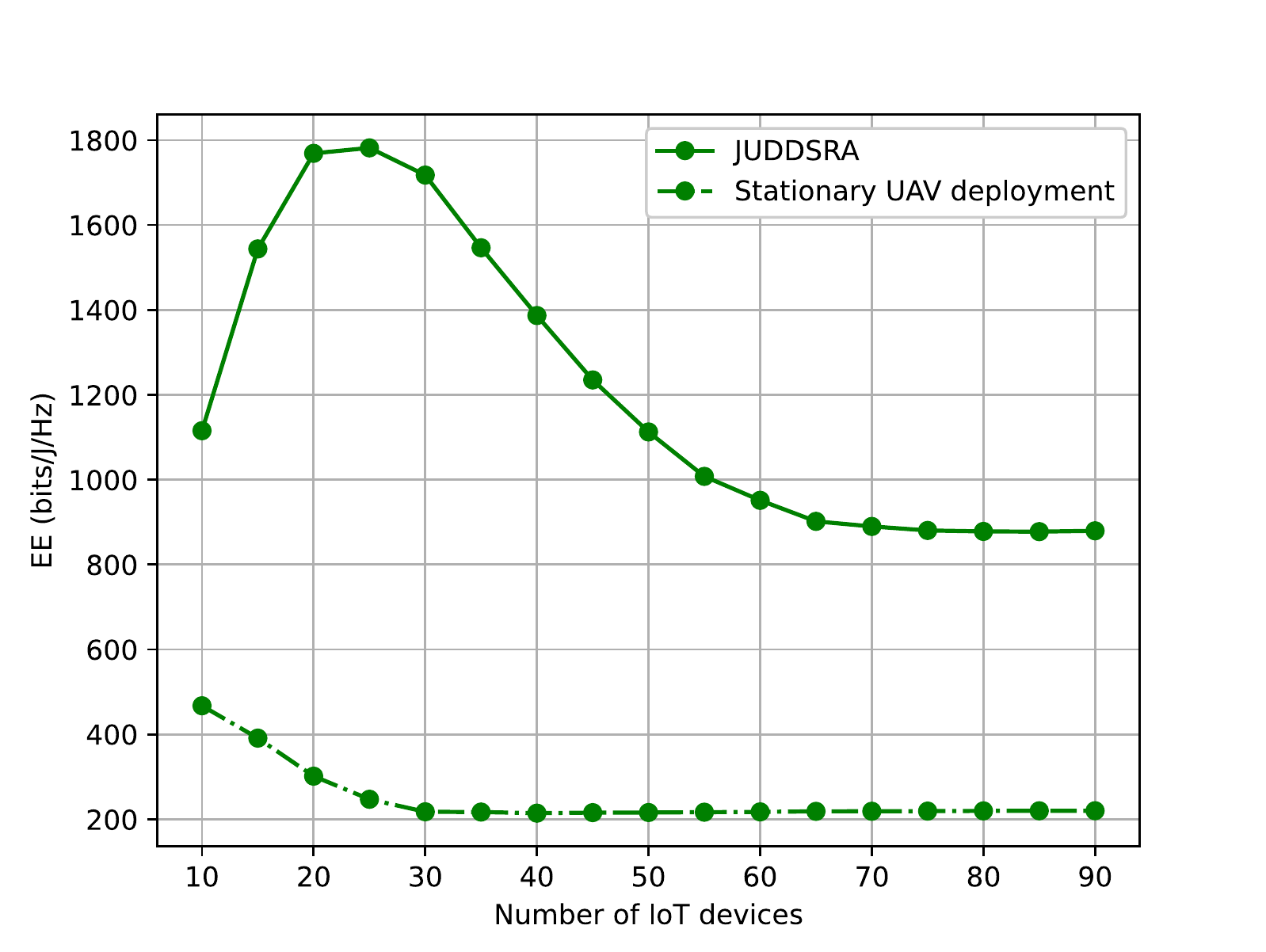}
  \caption{Comparison of average EE between JUDDSRA and stationary UAV deployment scenario.}
  \label{fig:EE_stationary_UAV}
\end{figure} 

Fig.~\ref{fig:EE_stationary_UAV} depicts the comparison of JUDDSRA with stationary UAV deployment case. For the stationary UAV deployment case, JDSSA-1 and DABPA are applied for the device scheduling, subchannel allocation and power allocation, respectively, where the UAV is assumed to be located at $(0, 0, H)$. We find out that JUDDSRA achieves much higher EE compared to the stationary UAV deployment case, which thus proves the importance of UAV deployment in UAV-aided networks.   
\section{Conclusions}
In this paper, we investigated the joint UAV deployment, device scheduling and resource allocation problem in UAV-aided IoT networks with NOMA, with the objective of maximizing network EE. For the device scheduling and subchannel allocation problem, we proposed a swap matching algorithm based on matching game. A novel concept of `exploration' was also introduced into the matching game to search for global optimal solutions. Properties of the proposed algorithms were analyzed in terms of stability, convergence, complexity and optimality. For power allocation, an efficient algorithm based on Dinkelbach's algorithm was proposed to obtain the optimal solution. By formulating the UAV deployment as a disk covering problem, we proposed a low-complexity approach to determine the minimum number of optimal stop points for UAV to collect data from ground devices via line-of-sight links. Numerical results showed that the proposed algorithm achieved much higher EE compared to the stationary UAV deployment and fixed power allocation cases. It was also shown that UAV-aided IoT networks with NOMA  outperformed the OMA case in terms of number of accessed devices.

\bibliographystyle{IEEEtran}
\bibliography{uav_noma_trans}

\begin{thebibliography}{10}
\providecommand{\url}[1]{#1}
\csname url@samestyle\endcsname
\providecommand{\newblock}{\relax}
\providecommand{\bibinfo}[2]{#2}
\providecommand{\BIBentrySTDinterwordspacing}{\spaceskip=0pt\relax}
\providecommand{\BIBentryALTinterwordstretchfactor}{4}
\providecommand{\BIBentryALTinterwordspacing}{\spaceskip=\fontdimen2\font plus
\BIBentryALTinterwordstretchfactor\fontdimen3\font minus
  \fontdimen4\font\relax}
\providecommand{\BIBforeignlanguage}[2]{{%
\expandafter\ifx\csname l@#1\endcsname\relax
\typeout{** WARNING: IEEEtran.bst: No hyphenation pattern has been}%
\typeout{** loaded for the language `#1'. Using the pattern for}%
\typeout{** the default language instead.}%
\else
\language=\csname l@#1\endcsname
\fi
#2}}
\providecommand{\BIBdecl}{\relax}
\BIBdecl

\bibitem{7470933}
Y.~{Zeng}, R.~{Zhang}, and T.~J. {Lim}, ``Wireless communications with unmanned
  aerial vehicles: opportunities and challenges,'' \emph{{IEEE} Commun. Mag.},
  vol.~54, no.~5, pp. 36--42, May 2016.

\bibitem{7317490}
L.~{Gupta}, R.~{Jain}, and G.~{Vaszkun}, ``Survey of important issues in {UAV}
  communication networks,'' \emph{{IEEE} Commun. Surveys Tuts.}, vol.~18,
  no.~2, pp. 1123--1152, Nov. 2015.

\bibitem{8432474}
Q.~{Wu}, J.~{Xu}, and R.~{Zhang}, ``Capacity characterization of {UAV}-enabled
  two-user broadcast channel,'' \emph{{IEEE} J. Sel. Areas Commun.}, vol.~36,
  no.~9, pp. 1955--1971, Sep. 2018.

\bibitem{7876852}
H.~{Menouar}, I.~{Guvenc}, K.~{Akkaya}, A.~S. {Uluagac}, A.~{Kadri}, and
  A.~{Tuncer}, ``{UAV}-enabled intelligent transportation systems for the smart
  city: Applications and challenges,'' \emph{{IEEE} Commun. Mag.}, vol.~55,
  no.~3, pp. 22--28, Mar. 2017.

\bibitem{5741148}
S.~{Lien}, K.~{Chen}, and Y.~{Lin}, ``Toward ubiquitous massive accesses in
  {3GPP} machine-to-machine communications,'' \emph{{IEEE} Commun. Mag.},
  vol.~49, no.~4, pp. 66--74, Apr. 2011.

\bibitem{6863654}
A.~{Al-Hourani}, S.~{Kandeepan}, and S.~{Lardner}, ``Optimal {LAP} altitude for
  maximum coverage,'' \emph{{IEEE} Wireless Commun. Lett.}, vol.~3, no.~6, pp.
  569--572, Dec. 2014.

\bibitem{8038869}
M.~{Mozaffari}, W.~{Saad}, M.~{Bennis}, and M.~{Debbah}, ``Mobile unmanned
  aerial vehicles ({UAVs}) for energy-efficient internet of things
  communications,'' \emph{{IEEE} Trans. Wireless Commun.}, vol.~16, no.~11, pp.
  7574--7589, Nov. 2017.

\bibitem{8247211}
Q.~{Wu}, Y.~{Zeng}, and R.~{Zhang}, ``Joint trajectory and communication design
  for multi-{UAV} enabled wireless networks,'' \emph{{IEEE} Trans. Wireless
  Commun.}, vol.~17, no.~3, pp. 2109--2121, Mar. 2018.

\bibitem{8974403}
D.~{Xu}, Y.~{Sun}, D.~W.~K. {Ng}, and R.~{Schober}, ``Multiuser {MISO} {UAV}
  communications in uncertain environments with no-fly zones: Robust trajectory
  and resource allocation design,'' \emph{IEEE Transactions on Communications},
  vol.~68, no.~5, pp. 3153--3172, May 2020.

\bibitem{7842433}
Z.~{Ding}, Y.~{Liu}, J.~{Choi}, Q.~{Sun}, M.~{Elkashlan}, I.~{Chih-Lin}, and
  H.~V. {Poor}, ``Application of non-orthogonal multiple access in {LTE} and
  {5G} networks,'' \emph{{IEEE} Commun. Mag.}, vol.~55, no.~2, pp. 185--191,
  Feb. 2017.

\bibitem{6868214}
Z.~{Ding}, Z.~{Yang}, P.~{Fan}, and H.~V. {Poor}, ``On the performance of
  non-orthogonal multiple access in {5G} systems with randomly deployed
  users,'' \emph{{IEEE} Signal Process. Lett.}, vol.~21, no.~12, pp.
  1501--1505, Dec. 2014.

\bibitem{ding2016}
Z.~{Ding}, P.~{Fan}, and H.~V. {Poor}, ``Impact of user pairing on {5G}
  nonorthogonal multiple-access downlink transmissions,'' \emph{{IEEE} Trans.
  Veh. Technol.}, vol.~65, no.~8, pp. 6010--6023, Aug. 2016.

\bibitem{lei2016power}
L.~Lei, D.~Yuan, C.~K. Ho, and S.~Sun, ``Power and channel allocation for
  non-orthogonal multiple access in {5G} systems: Tractability and
  computation,'' \emph{{IEEE} Trans. Wireless Commun.}, vol.~15, no.~12, pp.
  8580--8594, Dec. 2016.

\bibitem{7954630}
J.~{Zhao}, Y.~{Liu}, K.~K. {Chai}, A.~{Nallanathan}, Y.~{Chen}, and Z.~{Han},
  ``Spectrum allocation and power control for non-orthogonal multiple access in
  hetnets,'' \emph{IEEE Transactions on Wireless Communications}, vol.~16,
  no.~9, pp. 5825--5837, Sep. 2017.

\bibitem{8014491}
J.~{Zhao}, Y.~{Liu}, K.~K. {Chai}, Y.~{Chen}, and M.~{Elkashlan}, ``Joint
  subchannel and power allocation for noma enhanced {D2D} communications,''
  \emph{IEEE Transactions on Communications}, vol.~65, no.~11, pp. 5081--5094,
  Nov. 2017.

\bibitem{8663350}
X.~{Liu}, J.~{Wang}, N.~{Zhao}, Y.~{Chen}, S.~{Zhang}, Z.~{Ding}, and F.~R.
  {Yu}, ``Placement and power allocation for {NOMA}-{UAV} networks,''
  \emph{IEEE Wireless Communications Letters}, vol.~8, no.~3, pp. 965--968,
  Jun. 2019.

\bibitem{8482444}
T.~M. {Nguyen}, W.~{Ajib}, and C.~{Assi}, ``A novel cooperative noma for
  designing {UAV}-assisted wireless backhaul networks,'' \emph{IEEE Journal on
  Selected Areas in Communications}, vol.~36, no.~11, pp. 2497--2507, Nov.
  2018.

\bibitem{8918266}
R.~{Tang}, J.~{Cheng}, and Z.~{Cao}, ``Joint placement design, admission
  control, and power allocation for {NOMA}-based {UAV} systems,'' \emph{IEEE
  Wireless Communications Letters}, vol.~9, no.~3, pp. 385--388, Mar. 2020.

\bibitem{8685130}
F.~{Cui}, Y.~{Cai}, Z.~{Qin}, M.~{Zhao}, and G.~Y. {Li}, ``Multiple access for
  mobile-{UAV} enabled networks: Joint trajectory design and resource
  allocation,'' \emph{{IEEE} Trans. Commun.}, vol.~67, no.~7, pp. 4980--4994,
  Jul. 2019.

\bibitem{9123495}
Z.~{Na}, Y.~{Liu}, J.~{Shi}, C.~{Liu}, and Z.~{Gao}, ``{UAV}-supported
  clustered {NOMA} for {6G}-enabled internet of things: Trajectory planning and
  resource allocation,'' \emph{{IEEE} Internet Things J., Early Access}, 2020.

\bibitem{8700188}
R.~{Duan}, J.~{Wang}, C.~{Jiang}, H.~{Yao}, Y.~{Ren}, and Y.~{Qian}, ``Resource
  allocation for multi-{UAV} aided {IoT} {NOMA} uplink transmission systems,''
  \emph{IEEE Internet of Things Journal}, vol.~6, no.~4, pp. 7025--7037, Aug.
  2019.

\bibitem{9}
``{3GPP} {TR} 36.777: study on enhanced support for aerial vehicles,''
  \emph{[Online] Available: https://lnkd.in/gR5fpdf}.

\bibitem{10}
``Cellular drone communication: {LTE} drone trial report,'' \emph{[Online]
  Available:
  https://www.qualcomm.com/documents/lte-unmanned-aircraft-systems-trial-report}.

\bibitem{11}
X.~L. et~al., ``The sky is not the limit: {LTE} for unmanned aerial vehicles,''
  \emph{{IEEE} Commun. Mag.}, vol.~56, no.~4, pp. 204--210, Apr. 2018.

\bibitem{2}
A.~Al-Hourani, S.~Kandeepan, and S.~Lardner, ``Optimal {LAP} altitude for
  maximum coverage,'' \emph{{IEEE} Wireless Commun. Lett.}, vol.~3, no.~6, pp.
  569--572, Dec. 2014.

\bibitem{7542118}
Z.~{Yang}, Z.~{Ding}, P.~{Fan}, and N.~{Al-Dhahir}, ``A general power
  allocation scheme to guarantee quality of service in downlink and uplink
  {NOMA} systems,'' \emph{{IEEE} Trans. Wireless Commun.}, vol.~15, no.~11, pp.
  7244--7257, Nov. 2016.

\bibitem{5}
M.~Ismail, W.~Zhuang, E.~Serpedin, and K.~Qaraqe, ``A survey on green mobile
  networking: From the perspectives of network operators and mobile users,''
  \emph{{IEEE} Commun. Surveys Tuts.}, vol.~17, no.~3, pp. 1535--1556, Dec.
  2014.

\bibitem{roth1992two}
A.~E. Roth and M.~A.~O. Sotomayor, \emph{Two-sided matching: A study in
  game-theoretic modeling and analysis}.\hskip 1em plus 0.5em minus 0.4em\relax
  Cambridge University Press, 1992, no.~18.

\bibitem{gu2015future}
Y.~Gu, W.~Saad, M.~Bennis, M.~Debbah, and Z.~Han, ``Matching theory for future
  wireless networks: fundamentals and applications,'' \emph{{IEEE} Commun.
  Mag.}, vol.~53, no.~5, pp. 52--59, May 2015.

\bibitem{manlove2013algorithmics}
D.~F. Manlove, \emph{Algorithmics of matching under preferences}.\hskip 1em
  plus 0.5em minus 0.4em\relax World Scientific, 2013, vol.~2.

\bibitem{bodine2011peer}
E.~Bodine-Baron, C.~Lee, A.~Chong, B.~Hassibi, and A.~Wierman, ``Peer effects
  and stability in matching markets,'' in \emph{Algorithmic Game Theory}.\hskip
  1em plus 0.5em minus 0.4em\relax Springer, Apr. 2011, pp. 117--129.

\bibitem{7115904}
Y.~{Gu}, Y.~{Zhang}, M.~{Pan}, and Z.~{Han}, ``Matching and cheating in device
  to device communications underlying cellular networks,'' \emph{{IEEE} J. Sel.
  Areas Commun.}, vol.~33, no.~10, pp. 2156--2166, Oct. 2015.

\bibitem{6848847}
M.~{Hasan} and E.~{Hossain}, ``Distributed resource allocation for relay-aided
  device-to-device communication: A message passing approach,'' \emph{{IEEE}
  Trans. Wireless Commun.}, vol.~13, no.~11, pp. 6326--6341, Jul. 2014.

\bibitem{12}
A.~E. Roth, ``Deferred acceptance algorithms: history, theory, practice, and
  open questions,'' \emph{Int. J. Game Theory}, vol.~36, no. 3-4, pp. 537--569,
  Mar. 2008.

\bibitem{arnold2002dynamic}
T.~Arnold and U.~Schwalbe, ``Dynamic coalition formation and the core,''
  \emph{Journal of Economic Behavior \& Organization}, vol.~49, no.~3, pp.
  363--380, Nov. 2002.

\bibitem{3}
M.~Mozaffari, W.~Saad, M.~Bennis, and M.~Debbah, ``Unmanned aerial vehicle with
  underlaid device-to-device communications: Performance and tradeoffs,''
  \emph{{IEEE} Trans. Wireless Commun.}, vol.~15, no.~6, pp. 3949--3963, Jun.
  2016.

\end{thebibliography}

\end{document}